%% file: main.tex
\documentclass{LMCS}

\def\dOi{9(3:13)2013}
\lmcsheading%
{\dOi}
{1--29}
{}
{}
{Jan.~\phantom14, 2013}
{Sep.~11, 2013}
{}

\usepackage{amssymb,amsmath,amscd,amsfonts,latexsym}
\usepackage{tikz}
\usepackage{pgf}
\usetikzlibrary{arrows,automata,backgrounds,decorations}
\usetikzlibrary{shapes.multipart}
\usepgflibrary{decorations.pathreplacing} 
\usepackage{graphicx}
\usepackage{multirow}

\usepackage{enumerate}
\usepackage{hyperref}

\RequirePackage[T1]{fontenc}

\input{defs}

\begin{document}

\title{Reachability problem for weak multi-pushdown automata}

\author[W.~Czerwi{\'n}ski]{Wojciech Czerwi{\'n}ski\rsuper a}
\address{{\lsuper a}Universit\"at Bayreuth}
\email{wczerwin@mimuw.edu.pl}

\author[P.~Hofman]{Piotr Hofman\rsuper b}
\address{{\lsuper{b,c}}Institute of Informatics, University of Warsaw}
\email{\{ph209519,sl\}@mimuw.edu.pl}

\author[S.~Lasota]{S{\l}awomir Lasota\rsuper c}
%\address{Institute of Informatics, University of Warsaw}
%\email{sl@mimuw.edu.pl}

\thanks{{\lsuper a}The first author acknowledges a partial support by the Polish MNiSW grant N N206 568640.}
\thanks{{\lsuper{b,c}}The other authors acknowledge a partial support by the Polish MNiSW grant N N206 567840.}


\ACMCCS{[{\bf Theory of computation}]: Models of computation; Formal
  languages and automata theory--- Formalisms---Rewrite systems;
  Semantics and reasoning---Program reasoning---Program verification}

\begin{abstract}
This paper is about reachability analysis in a restricted subclass of
multi-pushdown automata.  We assume that the control states of an
automaton are partially ordered, and all transitions of an automaton
go downwards with respect to the order.  We prove decidability of the
reachability problem, and computability of the backward reachability
set.  As the main contribution, we identify relevant subclasses where
the reachability problem becomes NP-complete.  This matches the
complexity of the same problem for communication-free vector addition
systems, a special case of stateless multi-pushdown automata.
\end{abstract}

\maketitle

\section{Introduction}

This paper is about reachability analysis of \emph{multi-pushdown systems}, i.e.,
systems with a global control state and multiple stacks.
The motivation for our work is twofold.
On one side, a practical motivation coming from context-bounded analysis of recursive concurrent 
programs~\cite{QR05,LR09,ABQ11}.
On the other side, a theoretical motivation coming from partially-commutative context-free grammars, 
developed recently in~\cite{CFL09,CFL11,CL12}.\vfill

\smallsection{Context bounded analysis.}
Multi-pushdown systems  % , i.e., systems with a global state space and multiple stacks, are
may be used as an abstract model of concurrent programs with recursive procedures.
As multi-pushdown systems are a Turing-complete model of computation, they are only applicable for
verification under further tractable restrictions.
One remarkably successful restriction is imposing a bound on the number of context switches; 
between consecutive context switches,
the system may only perform operations on one stack (local operations).
In~\cite{QR05}, the context-bounded reachability has been shown decidable, 
by reduction to reachability of ordinary pushdown systems~\cite{BEM97}.
This line of research, with applications in formal verification, has been continued successfully, e.g., 
in~\cite{BESS06,LR09,ABQ11}.

\smallsection{Weak control states.}
As our starting point we observe that if the number of context switches is bounded,
one may safely assume that the control state space is \emph{weak}, in the sense that there is a partial order 
on control states such that transitions go only downwards with respect to the order.
Indeed, the local state space of every thread may be eliminated using a stack, and the global control
state essentially enumerates context switches.
Roughly speaking, the model investigated in this paper extends the above one with respect to operations allowed
between two context switches, namely, we do not restrict these operations to one stack only. 
Thus, if $k$ is the number of stacks, we assume that transitions of a system are of the following form: 
\begin{align}  \label{eq:trans-intro}
q, \ \s{X} \ \trans[a] \ q', \ \a_1, \ldots, \a_k,
\end{align}
to mean that in control state $q$, symbol $\s{X}$ is popped from one of the stacks, and 
sequences of symbols $\a_1$, \ldots, $\a_k$, respectively, are pushed on stacks.
Wlog. one may assume that the symbols of different stacks are different.

\smallsection{Partially-commutative context-free grammars.}
A special case of the model investigated in this paper is \emph{stateless} multi-pushdown systems.
This is still a quite expressible model that  subsumes, among the others, context-free graphs
(so called Basic Process Algebra~\cite{handbook}) and communication-free Petri nets
(so called Basic Parallel Processes~\cite{handbook}, or commutative context-free graphs).
In the stateless case, transitions~\eqref{eq:trans-intro} may be understood as productions of a grammar,
with the nonterminal symbols on the right-hand side (stack symbols) subject to a commutativity law.
More precisely, for any two symbols $\s{X}$ and $\s{Y}$ from different stacks, we impose
the commutativity law
$$\s{X} \s{Y} \ = \ \s{Y} \s{X} .$$
One easily observes that this is a special case of \emph{independence relation} over nonterminal symbols,
as defined in trace theory~\cite{bookoftraces}\footnote{Note however that the independence is imposed on nonterminal
symbols, and not on input letters, as usually in trace theory.}.
In multi-pushdown systems, the \emph{dependency relation} (the complement of independence relation)
is always transitive.
A general theory of context-free grammars modulo dependency relation that is not necessarily transitive, has been studied 
recently in~\cite{CL12}; complexity of bisimulation equivalence checking has been investigated
in~\cite{CFL09,CFL11}. The present paper complements these results by focusing on reachability analysis.

\smallsection{Contributions.}
We investigate the reachability problem for \weak\ multi-pushdown automata that asks, given two sets of configurations,
namely the source and target set, whether there is a sequence of transitions leading from some source configuration to
some target one. For finite representation of the source and target set, we restrict to regular sets of configurations only.
Roughly speaking, in this paper we present an almost complete map of decidability and complexity results, with respect
to natural restrictions of the source and target sets, and of the control states of multi-pushdown automata.

The paper contains two main results and a few accompanying ones.
As our first main result, we prove decidability of reachability for \weak\ multi-pushdown automata, in case when target set is a singleton.
Our argument is based on a suitable well order on the set of configurations, that strongly depends
on the assumption that the control states are weak. 

As the second main result, we identify additional restrictions under which the problem is
NP-complete. One such restriction is stateless multi-pushdown systems;
another restriction is that in every control state, the stacks may be emptied without changing the state.
Our result subsumes (and gives a simpler nondeterministic polynomial time algorithm for) the case of
communication-free Petri nets, assuming that the source and target markings of a net are given in unary.
Reachability of communication-free Petri nets is NP-complete as shown in~\cite{E97}, even when the source and target
markings are presented in binary.
NP-completeness of reachability for stateless multi-pushdown automata is similar to NP-completeness of the word problem for partially-commutative context-free grammars~\cite{H83},
where one asks if the given input word is accepted.
The reachability question is more difficult to answer, as no input word is given in advance.
In fact the main technical difficulty is to show the existence of a polynomial witness for reachability.

As further results, we investigate forward and backward reachability sets, and prove that
the backward reachability set of a regular set of configurations is regular and computable, while the forward
reachability set needs not be regular in general.
We also identify the decidability border for reachability of weak multi-pushdown systems.
Roughly speaking, the problem becomes undecidable when the target set of configurations is 
an unrestricted regular set, but is decidable if the target set is a singleton.

Finally, we prove a slightly surprising fact: reachability is decidable for multi-pushdown automata, with unrestricted control states, under the assumption that in every state, the stacks may be emptied without changing state.

The standard techniques useful for analysis of pushdown systems, such as pumping or the automaton-based approach 
of~\cite{BEM97}, do not extend to the multi-pushdown setting.
This is why the proofs of our results are based on new insights. The NP-membership proofs are, roughly speaking,
based on polynomial witnesses obtained by careful elimination of 'irrelevant' transitions.
On the other hand, the decidability results are based on a suitable well order on configurations.

\smallsection{Related research.}  % on multi-pushdown systems.}
%Ordinary pushdown systems admit efficient polynomial-time reachability analysis~\cite{BEM97}.
Multi-pushdown systems are a fundamental model of recursive multi-threaded programs.
This is why different instantiations of the multi-pushdown paradigm have been appearing in the literature recently,
most often in the context of formal verification.
We only mention here a few relevant positions we are aware of, without claiming completeness.
All the papers cited below bring some restricted decidability results for reachability or model checking.

Most often, a model has global control states, with stack operations subject to some restriction.
For instance, the author of~\cite{A10} assumes that the stacks are ordered, and pop operation can only be performed 
on the first nonempty stack.
Another example is the model introduced in~\cite{BMT05} and then further investigated e.g. in~\cite{BESS06,AB09,BE12},
that allows for unbounded creation of new stacks; on the other hand, operations on each stack are local,
thus no communication between threads is allowed.
%In a series of papers initiated by~\cite{QR05}, and then continued e.g. in~\cite{LR08,LR09,ABQ11},
%the principal assumption is that along a run, only a fixed number of context-switches between different
%threads is allowed.

Another possible approach is to replace global state space with some communication mechanism 
between threads.
Some successful results on analysis of multi-threaded programs communicating via locks,
in a restricted way, has been reported in~\cite{IGK05,K11,CMV12}.

In~\cite{LS02} the algorithm for reachability over PA~\cite{handbook} graphs has been provided.
The PA class is a generalization of both BPA and BPP that allows, similarly like multi-pushdown systems, 
both for sequential and concurrent (interleaved) behavior.
Finally, in~\cite{KRS09} the reachability problem has been shown decidable for Process Rewrite Systems~\cite{PRS}
extended with weak control states.

\smallsection{Outline.}
In Section~\ref{sec:prelim} we define the model we work with.
Then in Section~\ref{sec:reachability} we make explicit the variants the reachability problem we investigate in the paper,
and in Section~\ref{sec:summary} we state all our results. To make the presentation more concise and easier to assimilate,
the proofs of all the results outlined in Section~\ref{sec:summary} are moved to remaining sections, namely
Sections~\ref{sec:proofoflemma}--\ref{sec:undecid}.

In Section~\ref{sec:proofoflemma} we prove the fundamental property that the source set of configurations may be assumed to be a singleton without affecting complexity of the reachability problem. This section introduces also some basic terminology used in the
following sections.
The following three sections contain the proofs of the remaining results, grouped with respect to complexity.
In Section~\ref{sec:NP} we prove membership in NP, in Section~\ref{sec:decid} we show decidability, and Section~\ref{sec:undecid}
contains undecidability proofs.
In the last section we briefly discuss the cases where the complexity of the reachabiltiy problem is still open.

A preliminary version of this paper has appeared in~\cite{CHL12}.
This paper is an improvement of~\cite{CHL12}, extended with full proofs omitted there, and with
few new results (for instance Theorem~\ref{thm:decid_mpda} and Theorem~\ref{thm:reach-set}).
The content of this paper is included in the PhD thesis of the first author~\cite{CzPhD}.

%%% Local Variables: 
%%% TeX-master: "main"
%%% End: 

\section{Multi-pushdown automata}
\label{sec:prelim}

A multi-pushdown automaton (\MPDA) is like a single-pushdown one.
In a single step one symbol is popped from one of stacks\footnote{If we allowed for popping from more than one stack at a time,
the model would clearly become Turing-complete, even with one state only.}, and a number of symbols are
pushed on the stacks. 
Assume there are $k$ stacks. A transition of an automaton is thus of the form:
\begin{align} \label{eq:mpdatran}
q, \ \s{X} \ \trans[a] \ q', \ \a_1, \ldots, \a_k,
\end{align}
to mean that when an automaton reads $a$ in state $q$, it pops $\s{X}$ from one of the stacks, pushes the 
sequence of symbols $\a_i$ on the $i$th stack, for $i = 1 \ldots k$, and goes to state $q'$. 
We allow for silent transitions with $a = \eps$.
Observe that wlog. one may assume that stack alphabets are disjoint. 

Formally, the ingredients of an \MPDA\ are: a finite set of states $\al{Q}$,
% together with a distinguished initial state, 
the number of stacks $k$, pairwise-disjoint
finite stack alphabets $\al{S_1} \ldots \al{S_k}$, an input alphabet $\al{A}$, and a finite set
of transition rules:
\begin{equation} \label{eq:transrel}
\trans[] \ \ \subseteq \ \al{Q} \times (\bigcup_{i \leq k} \al{S_i}) \times (\al{A} \cup \{\eps\}) 
\times \al{Q} \times {\al{S_1}}^* \times \ldots \times {\al{S_k}}^*
\end{equation}
written as in~\eqref{eq:mpdatran}.
A configuration of an \MPDA\ is a tuple $\config{q}{\b_1}{\ldots}{\b_k} \in \al{Q} \times \al{S_1}^* \times \ldots \times \al{S_k}^*$.
The transition rules~\eqref{eq:mpdatran} induce the transition relation over all configurations in a standard way:
\[
\frac{q, \ \s{X} \ \trans[a] \ q', \ \a_1, \ldots, \a_k \ \ \ \ \ \ \ \s{X} \in \al{S_i} \ \ \ \ \ \ \ \b_i = \s{X} \b}
{\config{q}{\b_1}{\ldots \b_i \ldots}{\b_k} \trans[a] \config{q'}{\a_1 \b_1}{\ldots \a_i \b \ldots}{\a_k \b_k}}
\]
thus defining the configuration graph of an \MPDA.
For a configuration $\config{q}{\a_1}{\ldots}{\a_k}$, its \emph{size} is defined as 
the sum of lengths of the words $\a_i$, $i \leq k$.
The same applies to a right-hand side of any transition rule $q \ \s{X} \ \trans[a] \ q' \ \a_1 \ldots \a_k$.

An \MPDA\ is \emph{\stateless} if there is just one state (or equivalently no states).
Transition rules of an automaton are then of the form:
\begin{equation} \label{eq:statelessmpdatran}
\s{X} \ \trans[a] \ \a_1, \ldots, \a_k
\end{equation}
and configurations are of the form $\confless{\b_1}{\ldots}{\b_k}$. 

A less severe restriction on control states is the following one.
We say that an automaton is \emph{weak} if there is a partial order $\leq$ on its states 
such that every transition~\eqref{eq:mpdatran} satisfies $q'\leq q$. 
Clearly, every \stateless\ automaton is weak.

\begin{rem}
Note that \stateless\ one-stack automata are essentially context-free grammars in
Grei\-bach normal form. Thus the configuration graphs are precisely context-free graphs,
called also BPA graphs~\cite{PRS,handbook}.
Another special case is many stacks with singleton alphabets. The stacks are thus essentially
counters without zero tests. In this subclass, \stateless\ automata corresponds to communication-free Petri nets~\cite{E97},
called also BPP~\cite{ChrPhd}, or commutative context-free graphs~\cite{CFL11}.
The BPA and BPP classes are members of the Process Rewrite Systems hierarchy of~\cite{PRS} that contains,
among the others, unrestricted pushdown systems and Petri nets.
\end{rem}

\begin{exa}
\label{ex:automaton}
Assuming a distinguished initial state and acceptance by all stacks empty,
\weak\  \MPDA s\ can recognize non-context-free languages. 
For instance, the language 
\begin{align} \label{eq:langex}
\{ a^n b^n c^n : n \geq 0 \}
\end{align}
is recognized by an automaton described below.
The automaton has two states $q_1,q_2$ and two stacks. The alphabets of the stacks are $\{\s{X},\s{B}, \s{D}\}$ and $\{\s{C}\}$, respectively. 
% moreover there are three terminal symbols $\{a,b,c\}.$
The starting configuration is $(q_1, \s{X}\s{D}, \varepsilon)$.
% and the automaton accepts by a configuration $(q2, \varepsilon , \varepsilon).$
Besides the transition rules, we also present the automaton in a diagram, 
using \emph{push} and \emph{pop} operations with natural meaning.

\begin{minipage}{.55\linewidth}
%\begin{center}
                \begin{tikzpicture}    % [every text node part/.style={align=center}]                
	    	   \node (q1) at (0,0){$q_1$};
	    	   \node (q2) at (4, 0) {$q_2$};
	   	 
                    \path[->] (q1) edge [loop above] node [text width=1cm] {$a,\ pop\ \s{X} $ \\ $push\ \s{X}\s{B},\s{C}$} (q1);
		    \path[->] (q1) edge [loop left] node [text width=1cm] {$\varepsilon,$\\$pop\ \s{X} $} (q1);
		    \path[->] (q1) edge [loop below] node [text width=1cm] {$b,\ pop\ \s{B} $} (q1);
		    \path[->] (q2) edge [loop above] node [text width=1cm] {$c,\ pop\ \s{C} $} (q2);
		    \path[->] (q1) edge  node [above, text width=1cm] {$\varepsilon,\ pop\ \s{D} $} (q2);
               \end{tikzpicture}
%        \end{center}
\end{minipage}
\begin{minipage}{.45\linewidth}
\begin{align*}
q_1, \ \s{X}  \ & \trans[a] \  q_1, \ \s{X}\s{B}, \ \s{C}  \\
q_1, \ \s{X}  \ & \trans[\varepsilon] \  q_1, \ \varepsilon, \ \varepsilon  \\
q_1, \ \s{B}  \ & \trans[b]  \  q_1, \ \varepsilon, \ \varepsilon  \\
q_1, \ \s{D}  \ & \trans[\varepsilon]  \  q_2, \ \varepsilon, \ \varepsilon  \\
q_2, \ \s{C}  \ & \trans[c]  \  q_2, \ \varepsilon, \ \varepsilon 
\end{align*}
\end{minipage}
\ \\

\noindent
The automaton is weak and uses $\varepsilon$-transitions, which may be however easily eliminated.
Acceptance by empty stacks may be easily simulated using acceptance by states.
The language~\eqref{eq:langex} is not recognized by a \stateless\ automaton, as shown in~\cite{CL12}.
\end{exa}

\begin{exa}
%\label{ex:automaton}
Non-context-free languages are recognized even by \stateless\ \MPDAs\ with singleton stack alphabets.
The class of languages recognized by this subclass is called \emph{commutative context-free} languages~\cite{H83}, 
see also~\cite{CL12}.
One example is the commutative closure of the language of the previous example: the set of all words
with the same number of occurrences  of $a$, $b$ and $c$.

\end{exa}
\ \\
In the sequel we do not care about initial states nor about acceptance condition, as we will focus
on the configuration graph of an automaton. Furthermore, as we only consider reachability problems,
the labeling of transitions with input alphabet letters 
will be irrelevant, thus we write $\trans$ instead of $\trans[a]$ from now on. 

Using a standard terminology, we say that an \MPDA\ is \emph{\normed}
if for any state $q$ and any configuration $\config{q}{\a_1}{\ldots}{\a_k}$, 
there is a path to the empty configuration
\[
\config{q}{\a_1}{\ldots}{\a_k} \trans \ldots \trans \config{p}{\eps}{\ldots}{\eps}
\]
for some state $p$.
In general, whenever an \MPDA\ is not assumed to be \normed\ we call it \emph{\unnormed} for clarity.
Note that in all examples above the automata were \normed. In fact \normedness\ is not a restriction as far as languages
are considered.
 In the sequel we will however analyze the configuration graphs, and then \normedness\ will play a role.

Further, we say that an \MPDA\ is \emph{\strnormed}
if for any state $q$ and any configuration $\config{q}{\a_1}{\ldots}{\a_k}$, 
there is a path to the empty configuration
\[
\config{q}{\a_1}{\ldots}{\a_k} \trans \ldots \trans \config{q}{\eps}{\ldots}{\eps}
\]
containing only transitions that do not change state.
Intuitively, whatever is the state $q$ we start in, any top-most symbol $\s{X}$ in any stack may ,,disappear''.
For \stateless\ automata, \strnormedness\ is the same as \normedness.

\section{Regular sets and reachability problem}
\label{sec:reachability}

\smallsection{Regular sets}
We will consider various reachability problems in the configuration graph of a given \MPDA.
Therefore, we need a finite way of describing infinite sets of configurations.
A standard approach is to consider \emph{regular} sets. Below we adapt this approach
to the multi-stack scenario we deal with. Namely, instead of languages of words we consider
languages of tuples of words, one word per stack.

Consider the configurations of a \stateless\ \MPDA, $\confset = \al{S}_1^* \times \ldots \times \al{S}_k^*$.
There is a natural monoid structure in $\confset$, with pointwise identity
$\confless{\eps}{\ldots}{\eps}$ and multiplication
\[
\confless{\a_1}{\ldots}{\a_k} \cdot \confless{\b_1}{\ldots}{\b_k} \ = \ 
\confless{\a_1 \, \b_1}{\ldots}{\a_k \, \b_k} .
\]
Call a subset $L \subseteq \confset$ \emph{regular} if there is a finite monoid $M$ and a monoid morphism
$
\g : \confset \to M
$
that \emph{recognizes} $L$, which means that $L = \g^{-1}(N)$ for some subset $N \subseteq M$.
It is well-known (see for instance~\cite{BFL12} and references therein) that a language $L$ is regular iff it is a finite union of languages of the form
\[
L_1 \times \ldots \times L_k,
\]
where $L_i$ is a regular word language over $\al{S}_i$, for $i = 1 \ldots k$.
Thus one may assume, without loss of generality, that the monoid $M$ is a product
of finite monoids 
\[
M = M_1 \times \ldots \times M_k,
\] 
and that
\[
\gamma = \gamma_1 \times \ldots \times \gamma_k \quad
\text{ where} \quad \gamma_i : \al{S}_i^* \to M_i \ \text{ for } i = 1 \ldots k.
\]
Thus we may use an equivalent but more compact representation of regular sets,
based on automata: a regular set $L$ is given by a tuple of (nondeterministic) finite automata $\aut{B}_1 \ldots \aut{B}_k$ 
over alphabets $\al{S}_1 \ldots \al{S}_k$, respectively, together with a set 
\[
F \subseteq Q_1 \times \ldots \times Q_k
\]
of accepting tuples of states, where $Q_i$ denotes the state space of automaton $\aut{B}_i$.

Unless stated otherwise, in the sequel we always use such representations of regular sets
of configurations.  % as described above.
If there is more than one state, we assume a representation for every state.
In particular, 
when saying ''polynomial wrt.~$L$'', for a regular language $L$, we mean polynomial wrt.~the sum of sizes
of automata representing $L$.

\begin{rem}
Clearly, the cardinality of the set $F$ of accepting tuples may be exponential wrt.~the cardinalities of state spaces of
automata $\aut{B}_i$. However, complexities we derive in the sequel will never depend on cardinality of $F$.
\end{rem}
\begin{exa} \label{ex:reglang}
Assume that there are two stacks. An example of properties we can define is: 
,,odd number of elements on the first stack and symbol $\s{A}$ on the top of the second stack,
or an even number of the elements on the first stack and the odd number of elements on the second stack". 
On the other hand, ''all stacks have equal size" is not a regular property according to our definition.
\end{exa}

\begin{rem} \label{rem:reg}
We have deliberately chosen a notion of regularity of languages of \emph{tuples} of words.
Another possible approach could be to consider regular languages of words, over the product alphabet
$(\al{S}_1 \cup \bot) \times \ldots \times (\al{S}_k \cup \bot)$, where the additional symbol $\bot$
is necessary for padding. This would yield a larger class, for instance the last language from
Example~\ref{ex:reglang} would be regular.
The price to pay would be however undecidability of the reachability problems. 
The undecidability will be established in Theorem~\ref{thm:undecid-relaxed}.
\end{rem}

\begin{rem}
The notion of regularity for languages of tuples of words we work with is known in the literature also under the name
\emph{recognizable} languages, see for instance~\cite{BFL12, Berstel79}.
We prefer to stick to name \emph{regular}, in order to place emphasis on the fact that
the sets of configurations are assumed to be represented by tuples of finite automata.
\end{rem}

\smallsection{Reachability}
In this paper we consider the following reachability problem:\\

\begin{quote}
\begin{tabular}{ll}
% & {\sc Reachability problem} \\
%\\
{\sc Input}: \ & an \MPDA\ $\aut{A}$ and two regular sets of configurations $L, K \subseteq \confset$. \\
{\sc Question}: \ & is there a path in the configuration graph from $L$ to $K$?\\
\end{tabular}\\
\end{quote}

\noindent
We will write $L \reaches_{\aut{A}} K$ if a path from $L$ to $K$ exists in the configuration graph of an automaton $\aut{A}$.
The sets $L$ and $K$ we call \emph{source} and \emph{target} sets, respectively.
We will distinguish special cases, when either $L$ or $K$ or both the sets are singletons,
thus obtaining four different variants of reachability altogether. 
For brevity we will use symbol '\emph{\sing}' for a singleton, and symbol '\emph{\reg}'
for a regular set, and speak of \emph{\singreg\ reachability} (when $L$ is a singleton),
\emph{\regreg\ reachability} (the unrestricted case), and likewise for
\emph{\regsing} and \emph{\singsing}.

Before stating the results, we note that all the problems we consider here are \NP-hard:
\begin{lem} \label{lem:nphard}
The \singsing\ reachability is \NP-hard for \strnormed\ \stateless\ \MPDAs, even if all stack alphabets are singletons.
\end{lem}
The above fact follows immediately from NP-completeness of the reachability problem 
for communication-free Petri nets, see~\cite{E97} for details.

%%% Local Variables: 
%%% TeX-master: "main"
%%% End: 

\section{Summary of results}	
\label{sec:summary}

In this section we merely state all our results without any proofs. The proofs occupy all the following sections.
Our results come in three groups, each dealing with one of the following questions:
\begin{iteMize}{$\bullet$}
\item What is the complexity of the reachability problem in various variants mentioned above? Where the decidability border lies?
\item Is the set of forward/backward reachable configurations regular and computable?
\item What changes if the relaxed notion of regularity is assumed?
\end{iteMize}
The three groups are discussed in details in Sections~\ref{sec:summary-complexity}, \ref{sec:summary-reachableset} 
and~\ref{sec:summary-relaxed}, respectively.

Most of our results apply only to \weak\ \MPDAs, or even only to \stateless\ MPDAs.
There are however few exceptions: decidability of reachability (in Theorem~\ref{thm:decid_mpda}) and
regularity of backward reachability set (in Theorem~\ref{thm:reach-set}), where we consider unrestricted state spaces.
We want to emphasise that the results for unrestricted state spaces are rather accidental,
and are not the most interesting; the core part of the paper is devoted to \weak\ state spaces.

\subsection{Complexity of reachability}
\label{sec:summary-complexity}

In presence of states, the \singsing\ reachability problem is obviously undecidable, because the model
is Turing powerful. Undecidability holds even for \normed\ \MPDAs. The problem becomes
decidable under the assumption of \strnormedness:

\newcommand{\thmdecidmpda}{The \regreg\ reachability is decidable for \strnormed\ \MPDAs.}
\begin{thm}\label{thm:decid_mpda}
\thmdecidmpda
\end{thm}
The idea behind this result is that the model is similar to lossy FIFO systems, and hence is susceptible to the same proof methods. We exploit a well ordering of configurations, along the lines of~\cite{FS01}.

All our complexity results are outlined in Table~\ref{tab:summary}.
The table distinguishes
% Before stating the remaining results, we summarize all of them in the following table.
variants of the reachability problem for \strnormed/\normed/\unnormed\ \MPDAs, and restrictions on control states:  \stateless/\weak/unrestricted.
For clarity, we do not distinguish \stateless\ \strnormed\ case from \stateless\ \normed\ one, as these two cases 
obviously coincide.

{\small
\begin{table}[thb]
\newcommand{\ods}{\,}
\begin{tabular}{|c||c|c|c|}
% \cline{2-3} \multicolumn{2}{c|}{}                             \ods \strnormed        & \ods \unnormed \\
\hline
{\ods [ \regsing ] \ods } & \multirow{2}{*}{\ods \strnormed \ods}   & \multirow{2}{*}{\ods \ods \normed \ods \ods}   & \multirow{2}{*}{\ods \unnormed \ods} \\ 
\ods \regreg \ods &   & &  \\ 
\hline
\hline
\multirow{2}{*}{\ods \stateless \ods} & \multicolumn{2}{|c|}{\multirow{2}{*}{\NPc\ (Thm.~\ref{thm:np1})}}   & \ods [ \NPc\ (Thm.~\ref{thm:np2}) ] \ods \\ 
                                                         & \multicolumn{2}{|c|}{} & \ods \undecid\  (Thm.~\ref{thm:undecid}) \ods \\ 
\hline
\multirow{2}{*}{\ods \weak \ods}      & \ods \multirow{2}{*}{\NPc\  (Thm.~\ref{thm:np1})} \ods  &
          [ {\decid}\ ] & [ {\decid}\ (Thm.~\ref{thm:decid}) ] \\
                                                        &  & \ods \undecid\  (Thm.~\ref{thm:undecid}) \ods & \undecid\  \\
\hline
\multirow{2}{*}{\ods unrestrictred \ods}    &  \ods \multirow{2}{*}{\decid\ (Thm.~\ref{thm:decid_mpda})}  \ods &   \ods \multirow{2}{*}{\undecid}  &  \ods \multirow{2}{*}{\undecid}  \ods \\
&&& \\
\hline 
\end{tabular}
\caption{Summary of complexity results}
\label{tab:summary}
\end{table}
}
% \bigskip

The remainder of Section~\ref{sec:summary-complexity} is devoted to \weak\ state spaces only.
We start by observing that out of four combinations of the reachability problem,
it is sufficient to consider only two, namely the \regsing\ and \regreg\ cases. Indeed, as far as complexity is
concerned, we observe the following collapse:
\begin{align} \label{eq:collapse}
\text{ \singsing } & \ = \  \text{ \regsing } &
\text{ \singreg } & \ = \  \text{ \regreg }
\end{align}
independently of a restriction on automata.
The first equality follows from our first result, that says that for a given target configuration, there is a source configuration, witnessing  reachability, of polynomial size:
\newcommand{\lemreduceproblems}{
Suppose $\aut{A}$ is a \weak\ \unnormed\ \MPDA. Let $L$ be a regular set of configurations of $\aut{A}$ and let
$\c{t}$ be a configuration of $\aut{A}$. Then
\[
L \reaches_\aut{A} \c{t} \  \implies \ 
\c{s} \reaches_\aut{A} \c{t} \ \text{ for some } \c{s} \in L \text{ of size polynomial wrt.}~\aut{A}, L \text{ and } \c{t} .
\]}
\begin{lem}[Polynomial source configuration] \label{lem:reduceproblems}
\lemreduceproblems
\end{lem}

Indeed, the reduction from \regsing\ to \singsing\ is by nondeterministic guessing of a source configuration of polynomial size.
The second equality~\eqref{eq:collapse} will follow from our results listed below.

The collapse~\eqref{eq:collapse} significantly simplifies the summary of results in Table~\ref{tab:summary}.
Each entry of the table contains the complexity of \regreg\ reachability problem.
Additionally, the complexity of \regsing\ reachability problem is given in cases it is different
from the complexity of \regreg\ reachability.
%Note that the source set is never restricted, and it turns out that restricting the source set
%does not reduce complexity of reachability.
%

%
%Moreover, we do not distinguish \weak\ \normed\ case from the \weak\ \unnormed\ one, as they are easily
%interreducible (cf.~the proof of Theorem~\ref{thm:undecid}).
%

\indent
Now we discuss the results in detail.  % ok \todo{maybe details}.
We first observe an apparent decidability frontier witnessed
by \stateless\ \unnormed\ \MPDAs\ and \weak\ \normed\ \MPDA s:
\newcommand{\thmundecid}{The \singreg\ reachability is undecidable for \stateless\ \unnormed\ \MPDAs\ and for \weak\ \normed\ \MPDA s.}
\begin{thm} \label{thm:undecid}
\thmundecid
\end{thm} 
The proof is by reduction of the nonemptiness of intersection of context-free languages
and uses three stacks. The case of two stacks remains open.

Thus lack of \strnormedness\ combined with a regular target set yields undecidability in case of 
\weak\ automata.
Surprisingly, restricting additionally:
\begin{iteMize}{$\bullet$}
\item either the automaton to be \strnormed,
\item or the target set to a singleton, 
\end{iteMize}
makes a dramatical difference for complexity of the problem,
as summarized in Theorems~\ref{thm:np1}, \ref{thm:np2} and~\ref{thm:decid} below. 
In the first theorem we only assume \strnormedness:
\newcommand{\thmnpone}{The \regreg\ reachability is \NP-complete for \strnormed\ \weak\ \MPDAs.}
\begin{thm} \label{thm:np1}
\thmnpone
\end{thm}
Theorem~\ref{thm:np1} is the main result of this paper.
It is proved by showing that reachability is always witnessed by a polynomial witness,
obtained by careful elimination of 'irrelevant' transitions.

In the following two theorems we do not assume  \strnormedness, thus according to Theorem~\ref{thm:undecid}
we have to restrict target set to singleton. 
Under such a restriction, we are able to prove \NP-completeness only in the class of \stateless\ \MPDAs,
while for all  \weak\ \MPDAs\ we merely state decidability:

\newcommand{\thmnptwo}{The \regsing\ reachability is \NP-complete for \stateless\ \unnormed\ \MPDAs. }
\begin{thm} \label{thm:np2}
\thmnptwo
\end{thm}
\newcommand{\thmdecid}{The \regsing\ reachability is decidable for \weak\ \unnormed\ \MPDAs.}
\begin{thm} \label{thm:decid}
\thmdecid
\end{thm}
%
%Concerning the proofs, 
Theorem~\ref{thm:np2} is shown similarly to Theorem~\ref{thm:np1},
while the proof of Theorem~\ref{thm:decid} is based on a well order, the point-wise extension of a variant of Higman ordering.

\smallsection{Open questions}
Except for two few entries in the summarizing table above, where we merely prove decidability, 
we know the exact complexity of the reachability problem.
The remaining open question are summarized in Section~\ref{sec:conc}.
%that remains is the actual complexity of \singsing\ reachability for (\normed\ and \unnormed) \weak\ \MPDAs.
%Another interesting question is whether undecidability carries over to automata with two stacks only.

\subsection{Reachability set}  \label{sec:summary-reachableset}

Now we consider the problem of computing the whole reachability set.
For a given automaton $\aut{A}$, and a set $L$ of configurations, we consider forward and backward reachability sets of $L$, defined as:
\begin{align*}
\{ s : L \reaches_\aut{A} s \} \qquad \text{and} \qquad \{ s : s \reaches_\aut{A} L \},
\end{align*}
respectively.
It turns out that the backward reachability set is regular, whenever $L$ is regular,
under the \strnormedness\ assumption, even for unrestricted state spaces.

\newcommand{\thmreachset}{For \strnormed\ \MPDAs, 
the backward reachability set of a regular set is regular.
%For \weak\ \strnormed\ \MPDAs, the backward reachability set is effectively computable.
}
\begin{thm} \label{thm:reach-set}
\thmreachset
\end{thm}
Roughly speaking, we show that the backward reachability set is upward closed with respect to the point-wise extension of a suitable variant of Higman ordering. Note that we do not claim however that the reachability set may be effectively computed.

On the other hand, the forward reachability set needs not be regular, even in the case of \strnormed\ \stateless\ automata,
as shown in the following example.

\begin{exa}
Consider a \strnormed\ \stateless\ automaton with two stacks, over alphabets $\{\s{A}, \s{X}\}$ and $\{\s{B}\}$, and the following 
transition rules:
\[
\s{X} \trans \s{X} \s{A}, \  \s{B} \qquad \s{X} \to \varepsilon, \ \varepsilon \qquad \s{A} \to \varepsilon,\ \varepsilon \qquad \s{B} \to \varepsilon, \ \varepsilon.
\]
The set of configurations reachable from the configuration $(\s{X}, \varepsilon)$ is not regular:
\[ 
\{ (\s{A}^i, \s{B}^j) : i, j \in \Nat \} \ \ \cup \ \  \{ (\s{X} \s{A}^k, \s{B}^l) : k \geq l \} .
\]
%which is not a regular set.
\end{exa}

\subsection{Relaxed regularity}  \label{sec:summary-relaxed}

The relaxed definition of regularity,  as discusses in Remark~\ref{rem:reg}, 
makes the reachability problem intractable in all cases. The following theorem is shown
by reduction from the Post Correspondence Problem:

\newcommand{\thmundecidrelaxed}{The \singreg\ reachability is undecidable for \stateless\ \strnormed\ \MPDAs, under the relaxed notion of regularity.}
\begin{thm} \label{thm:undecid-relaxed}
\thmundecidrelaxed
\end{thm}
Furthermore, the backward reachability set of a relaxed regular set is not necessarily regular, 
even in \stateless\ \strnormed\ \MPDAs, as illustrated by the following example.

\begin{exa}\label{example_backward_set}
The automaton uses two stacks, with alphabets $\{\s{A},\s{X}, \s{B}\}$ and $\{\s{C}\}$. 
Every symbol has a disappearing rule:
%\begin{align*}
$\s{A} \trans \varepsilon, \ \varepsilon$,
% &&
%\s{X} \trans \varepsilon, \ \varepsilon &&
%\s{B} \trans \varepsilon, \ \varepsilon &&
%\s{C} \trans \varepsilon, \ \varepsilon 
%\end{align*}
and likewise for $\s{X}$, $\s{B}$ and $\s{C}$.
Additionally there is a transition rule $\s{B} \trans \s{C}$.
Consider the relaxed regular language
\[
L = \{ (\s{X} \s{A}^n, \s{C}^n) : n \geq 0 \} 
\]
and its backward reachability set. Denote by $K$ the subset of the backward reachability set
consisting of configurations with the second stack empty.
We claim that the projection of $K$ on the alphabet of the first stack is a word language which is not regular.
Indeed, as the only non-disappearing rule is $\s{B} \trans \s{C}$, configurations from $K$
have on the first stack a word of the form
$
w \s{X} \s{A}^n,
$
with at least $n$ occurrences of $\s{B}$ in $w$.
\end{exa}

The remaining sections contain proofs of all the results announced Section~\ref{sec:summary}. 
For Reader's convenience, every result is restated before its proof.

%%% Local Variables: 
%%% TeX-master: "main"
%%% End: 

\section{Proof of the polynomial source configuration property}   % Lemma~\ref{lem:reduceproblems}}
\label{sec:proofoflemma}

\restate{Lemma~\ref{lem:reduceproblems}}{\lemreduceproblems}

\begin{proof}
Consider an \MPDA\ $\aut{A}$ and a regular set $L$ of configurations of $\aut{A}$.
Let $\c{s} \in L$ be source configuration and let $\c{t}$ be an arbitrary target configuration. 
Suppose $\c{s} \reaches_\aut{A} \c{t}$.
We will show that the size of $\c{s}$ may be reduced, while preserving membership in $L$.
The crucial but simple idea of the proof will rely on an analysis of \emph{relevance} of 
symbol occurrences, to be defined below. 

\smallsection{Symbol occurrences.}
Suppose that there is a path $\pi$ from $\c{s}$ to $\c{t}$, consisting of consecutive transitions
$
\c{s}\ \trans \ \c{s}_1 \ \trans \ \c{s}_2 \ \ldots \ \trans \ \c{s_n} = \c{t} .
$
We will consider all individual occurrences of symbols that appear in the configurations.
For instance, in the following exemplary sequence of two-stack configurations
\begin{align} \label{eq:config}
\langle q, \s{A}\s{A}, \s{C} \rangle
\ \trans \ 
\langle q, \s{B}\s{B}\s{A}, \s{D}\s{C} \rangle
\ \trans \ 
\langle q, \s{A}\s{B}\s{B}\s{A}, \s{D}\s{C} \rangle
\end{align}
there are altogether 14 \emph{symbol occurrences}: 3 in the first configuration, 5 in the second one
and 6 in the third one.

Recall that every transition 
%\begin{align} % \label{eq:trans}
$\c{s}_i \trans \c{s}_{i+1}$
%\end{align}
is induced by some transition rule $q_1,\ \s{X} \trans q_2,\ \a$ of the automaton. 
Then there is a distinguished occurrence of symbol $\s{X}$ in $\c{s}_i$ 
that is involved in the transition. In the sequel we use the term \emph{symbol occurrence involved in a transition}.
%~\eqref{eq:trans}}.

Precisely one occurrence of symbol in $\c{s}_i$ is involved in the transition $\c{s}_i \trans \c{s}_{i+1}$; 
for every other occurrence of a symbol in $\c{s}_i$ there is a \emph{corresponding} 
occurrence of the same symbol in $\c{s}_{i+1}$.
(Note that we always make a difference between corresponding symbol occurrences from different configurations.)
All remaining occurrences of symbols in $\c{s}_{i+1}$ are created by the transition; 
we call these occurrences \emph{fresh}.

We define the \emph{descendant} relation as follows.
All fresh symbol occurrences in $\c{s}_{i+1}$ are descendants of the symbol occurrence in $\c{s}_i$ involved in the transition
$\c{s}_i \trans \c{s}_{i+1}$. 
Moreover, a symbol occurrence in $\c{s}_{i+1}$ corresponding to a symbol occurrence in $\c{s}_i$ is its descendant too.
We will use term \emph{descendant} for the reflexive-transitive closure of the relation defined above
and the term \emph{ancestor} for its inverse relation.
In particular, every symbol occurrence in $\c{t}$ is descendant of a unique symbol occurrence in $\c{s}$.
The descendant relation is a forest, i.e., a disjoint union of trees.

\begin{exa}
As an example, consider again the sequence of  transitions~\eqref{eq:config}, with 
symbol occurrences identified by subscripts $1 \ldots 14$: 
\begin{align} \label{eq:configoccurr}
\langle q, \s{A}_{1}\s{A}_2, \s{C}_3 \rangle
\ \trans \ 
\langle q, \s{B}_4\s{B}_5\s{A}_6, \s{D}_7\s{C}_8 \rangle
\ \trans \ 
\langle q, \s{A}_9\s{B}_{10}\s{B}_{11}\s{A}_{12}, \s{D}_{13}\s{C}_{14} \rangle
\end{align}
Say the transitions are induced by the following two transition rules:
\begin{align*}
q, \ \s{A} \ & \trans  \ q, \ \s{B}\s{B}, \  \s{D}  &
q, \ \s{D} \ & \trans  \ q, \ \s{A}, \  \s{D}
\end{align*}
The descendant relation can be presented as the following forest:
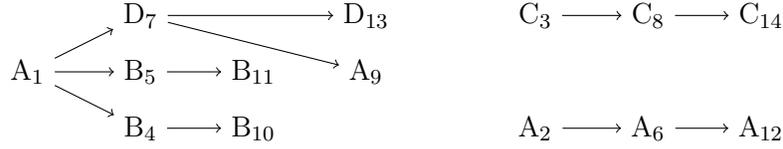
\begin{center} \label{eq:desc-tree}
  \begin{tikzpicture}[scale=0.75]             
	    	   \node (A1) at (0,1)	{$\s{A}_1$};
		   \node (B4) at (2,0)	{$\s{B}_4$};
	    	   \node (B5) at (2, 1) {$\s{B}_5$};
	    	   \node (D7) at (2, 2) {$\s{D}_7$};
		   \node (B10) at (4, 0) {$\s{B}_{10}$};
		   \node (B11) at (4, 1) {$\s{B}_{11}$};
		   \node (A9) at (6, 1)	{$\s{A}_9$};
	    	   \node (D13) at (6, 2) {$\s{D}_{13}$};
		   
		   \node (C3) at (0+9, 2) {$\s{C}_3$};
		   \node (C8) at (2+9, 2) {$\s{C}_8$};
		   \node (C14) at (4+9, 2) {$\s{C}_{14}$}; 

		   \node (A2) at (0+9, 0) {$\s{A}_2$};
		   \node (A6) at (2+9, 0) {$\s{A}_6$};
		   \node (A12) at (4+9, 0)  {$\s{A}_{12}$};
		    
		    \path[->] (A1) edge (B4);
		    \path[->] (A1) edge (B5);
		    \path[->] (A1) edge (D7);
		    \path[->] (B4) edge (B10);
		    \path[->] (B5) edge (B11);
		    \path[->] (D7) edge (A9);
		    \path[->] (D7) edge (D13);
		
		    \path[->] (C3) edge (C8);
		    \path[->] (C8) edge (C14);

		    \path[->] (A2) edge (A6);
		    \path[->] (A6) edge (A12);

               \end{tikzpicture}
\end{center}
\end{exa}
\noindent
The symbol occurrences involved in the two transitions~\eqref{eq:configoccurr}
are $\s{A}_1$ in the first configuration and $\s{D}_7$ in the second one.
The fresh symbol occurrences are $\s{B}_4$, $\s{B}_5$ and $\s{D}_7$ in the second configuration,
and $\s{A}_{9}$ and $\s{D}_{13}$ in the third one.

\smallsection{Relevant symbol occurrences.}
% Suppose that there is a path $\pi$ from $\c{s}$ to $\c{t}$.
As the automaton $\aut{A}$ is weak, the number of transitions in $\pi$ that change state is bounded
by the number of states of $\aut{A}$. All remaining transitions in $\pi$ do not change state.

Consider all the occurrences of all symbols in all configurations along the path $\pi$, 
including configurations $\c{s}$ and $\c{t}$ themselves.
A symbol occurrence is called \emph{relevant} if some of its descendants:
\begin{iteMize}{$\bullet$}
\item belongs to the target configuration $\c{t}$; or
\item is involved in some transition in $\pi$ that changes state.
\end{iteMize}
Otherwise, a symbol occurrence is \emph{irrelevant}.
In particular, all symbol occurrences in $\c{t}$ are relevant.
Referring back to our example, all symbol occurrences appearing in~\eqref{eq:configoccurr} are relevant.

Note that if $t$ is not the empty configuration then every configuration in $\pi$ contains at least one relevant symbol occurrence.
On the other side, in every configuration, the number of relevant occurrences is always bounded
by the sum of the size of $\c{t}$ and the number of states of $\aut{A}$.

\smallsection{Small source configuration.}
So prepared, we are ready to prove that there is a configuration $\c{s}' \in L$ of polynomial size with 
$\c{s}' \reaches_\aut{A} \c{t}$. We will rely on the following lemma:
\begin{lem} \label{lem:order}
For any configuration $\c{s}'$ obtained from $\c{s}$ by removing some irrelevant symbol occurrences,
it holds $\c{s}' \reaches_\aut{A} \c{t}$.
\end{lem}
The lemma follows from the following two observations: (1)
all the transitions in $\pi$ involving symbol occurrences remaining in $\c{s}'$
and their descendants may be re-done;
(2) the resulting configuration will be exactly $\c{t}$, as 
only irrelevant symbol occurrences have been removed from $\c{s}$.

Recall that the language $L$ is represented by a tuple $\aut{B}_1 \ldots \aut{B}_k$ of deterministic finite automata, 
one automaton per stack. 
Consider the content of a fixed $i$th stack in $\c{s}$, say $w \in \al{A}_i^*$. 
Let $n$ be the number of states of $\aut{B}_i$.
The run of the automaton $\aut{B}_i$ over $w$ labels every position of $w$ by some state. 
We will use a standard pumping argument to argue that every block of consecutive irrelevant symbol occurrences
in $\c{s}$ may be reduced in length to at most $n$. Indeed, upon every repetition of a state of $\aut{B}_i$,
the word $w$ may be shortened by removing the induced infix, while preserving membership in $L$.
By repeating the pumping argument for all blocks of consecutive irrelevant symbol occurrences in 
all stacks in $\c{s}$, one obtains a configuration $\c{s}'$, still belonging to $L$,
of quadratic size.
By Lemma~\ref{lem:order} we know that $\c{s}' \reaches \c{t}$, as required.
\end{proof}

\section{NP-completeness}
\label{sec:NP}

In this section we prove the following two theorems of Section
\ref{sec:summary}:\bigskip

\restate{Theorem~\ref{thm:np1}}{\thmnpone} 

\medskip

\restate{Theorem~\ref{thm:np2}}{\thmnptwo} 

\bigskip

NP-hardness in both cases follows from Lemma~\ref{lem:nphard}.
Concerning Theorem~\ref{thm:np1}, the proof of membership in \NP\ relies on the following two core lemmas:

\begin{lem} \label{lem:np1-proof-step1}
The \singsing\ reachability problem is in \NP\ for \strnormed\ \weak\ \MPDAs.
\end{lem}

\begin{lem} \label{lem:np1-proof-step2}
Let $\aut{A}$ be a \strnormed\ \weak\ \MPDA\ and let $L, K$ be regular sets of configurations.
If $L \reaches K$ then $\c{s} \reaches \c{t}$ for some $\c{s} \in L$ and $\c{t} \in K$ of size polynomial wrt.~the sizes of $\aut{A}$, $L$ and $K$. 
\end{lem}
Indeed, the two lemmas easily yield a decision procedure for \regreg\ reachability:
guess configurations $\c{s} \in L$ and $\c{t} \in K$ of size bounded by a polynomial deduced from
the proof of Lemma~\ref{lem:np1-proof-step2}, and then apply the procedure of Lemma~\ref{lem:np1-proof-step1}
to check if $\c{s} \reaches \c{t}$.

Concerning Theorem~\ref{thm:np2}, membership in \NP\ follows easily by Lemma~\ref{lem:reduceproblems} together
with the following:
\begin{lem} \label{lem:np2-proof}
The \singsing\ reachability problem is in \NP\ for \stateless\ \unnormed\ \MPDAs.
\end{lem}

The proof of Lemma~\ref{lem:np2-proof} is similar to the proof of Lemma~\ref{lem:np1-proof-step1}, as 
the irrelevant symbol occurrences must necessarily be normed, because they do not contribute
to the target configuration. We thus skip this proof, and devote 
the rest of this section to the proofs of Lemmas~\ref{lem:np1-proof-step1} and~\ref{lem:np1-proof-step2}.

\begin{proofof}{Lemma~\ref{lem:np1-proof-step1}}
Consider an \MPDA\ $\aut{A}$ and two configurations $\c{s}$ and $\c{t}$.  %, the source and target one, respectively.
We will define a nondeterministic polynomial-time decision procedure for $\c{s} \reaches_{\aut{A}} \c{t}$.

\smallsection{Stateless assumption.}
For simplicity, we assume that both $\c{s}$ and $\c{t}$ have the same control state.
Thus we can treat transitions that lead from $\c{s}$ to $\c{t}$ as \stateless\ transitions. 
At the very end of the proof, we will discuss how to generalize it to the general case of \strnormed\ 
\weak\ \MPDAs.

\smallsection{Polynomial witness.}
Our aim is to show that if there is a path from $\c{s}$ to $\c{t}$ then there is a path of polynomial length.
So stated, the above claim may not be verbally true, even in the case of context-free graphs,
as witnessed by the following simple example.
\begin{align}  \label{eq:exponen}
\s{X}_1 & \trans \s{X}_2 \s{X}_2 & 
\s{X}_2 & \trans \s{X}_3 \s{X}_3 & 
& \ldots &
\s{X}_{n-1} & \trans \s{X}_n \s{X}_n &
\s{X}_n & \trans \varepsilon 
\end{align}
The example scales with respect to $n$, and thus the shortest path from 
the configuration $\s{X}_1$ to $\s{X}_n$ is of exponential length.
As a conclusion, one must use some subtle analysis in order to be able to reduce the length of a witness of existence of the path as required.
Note that $\s{X}_1$ is relevant 
and thus can not be simply omitted.

\smallsection{Proof idea.} 
As a first step towards a polynomial bound on the witness of the path from $\c{s}$ to $\c{t}$, we will 
modify the notion of transition. Intuitively speaking, our aim is to consider exclusively relevant symbol occurrences.

By a \emph{subword} we mean any subsequence of a given word. For instance, $aaccbc$ is a subword of $aacabbcbcbc$. 
Further, by a \emph{subtransition} of $\s{X} \ \trans \ \a_1 \ldots \a_k$ we mean any 
$\s{X} \ \trans \ \b_1 \ldots \b_k$ such that the following conditions hold:
\begin{iteMize}{$\bullet$}
\item \emph{subword}: $\b_i$ is a subword of $\a_i$, for all $i \in \{1 \ldots k\}$; and
\item \emph{nonemptiness}: $\b_1 \ldots \b_k \neq \varepsilon$, i.e., at least one of words $\b_i$ is nonempty.
\end{iteMize}
Note that relying on the notion of relevance one easily deduces that
whenever there is a sequence of transitions from $\c{s}$ to $\c{t}$, 
then there is also  sequence of subtransitions.
Indeed, it is sufficient to remove irrelevant symbol occurrences in all transitions along the path from $\c{s}$ to $\c{t}$.

Clearly, the converse implication is not true in general.
For instance, if we add symbols $\s{X}_0$, $\s{A}$ and transitions
$\s{X}_0 \trans \s{X}_1 \s{A}$, $\s{A} \trans \varepsilon$
to the example~\eqref{eq:exponen}, there is a sequence of subtransitions
from the configuration $\s{X}_0$ to $\s{X}_n$, but no sequence of transitions from $\s{X}_0$ to $\s{X}_n$.
Our aim now it to modify the notion of subtransition in such a way that the converse implication does hold as well,
i.e., that existence of a sequence of subtransitions implies existence of a sequence of transitions. 
This requires certain amount of laborious book-keeping, as defined in detail below.

\smallsection{Marked subtransitions.}
We will need an additional copy of every stack alphabet $\al{A}_i$, denoted by $\bar{\al{A}}_i$,
for $i = 1 \ldots k$.
Thus for every $a \in \al{A}_i$ there is a corresponding marked symbol $\bar{a} \in \bar{\al{A}}_i$.
Formally, let the $i$th stack alphabet be $\al{A}_i \cup \bar{\al{A}}_i$.

A \emph{marked subword} of a word $w \in \al{A}^*_i$ is any word in $(\al{A}_i \cup \bar{\al{A}}_i)^*$
that may be obtained from $w$ by the following \emph{marking procedure}:
\begin{iteMize}{$\bullet$}
\item color arbitrary occurrences in $w$ (the idea is to color irrelevant symbol occurrences);
\item choose a prefix of $w$ and mark every occurrence in this prefix; the prefix should contain all occurences that are followed by some colored occurrence;
\item and finally remove colored occurrences.
\end{iteMize}
For instance, the following four words are all the marked subwords of $\mathrm{AACABBCBCBC}$, 
with respect to the coloring $\mathrm{A\mathbf{AC}A\mathbf{BB}C\mathbf{B}CBC}$:
\begin{align}  \label{eq:markmore}
\mathrm{\bar{A}\bar{A}\bar{C}CBC} \qquad
\mathrm{\bar{A}\bar{A}\bar{C}\bar{C}BC} \qquad
\mathrm{\bar{A}\bar{A}\bar{C}\bar{C}\bar{B}C} \qquad
\mathrm{\bar{A}\bar{A}\bar{C}\bar{C}\bar{B}\bar{C}}.
\end{align}
Recall that a word $w \in \al{A}_i^*$ represents a content of the $i$th stack, with the left-most symbol being the top-most.
Intuitively, the idea behind the notion of marked subword is to keep track of removed occurrences that are
covered by other symbols on the stack. For technical reasons in the second point above we allow 
additionally marking some further symbols that are not followed by a colored occurrence.

A notion of \emph{marked subtransition} is a natural adaptation of the notion of subtransition.
Compared to subtransitions, there are two differences:
'subword' is replaced with 'marked subword'; and whenever the left-side symbol is marked,
then it may only put marked symbols on its stack.
Formally, a marked subtransition of $\s{X} \ \trans \ \a_1 \ldots \a_k$ is any $\s{X} \ \trans \ \b_1 \ldots \b_k$ such that
the following conditions hold:
\begin{iteMize}{$\bullet$}
\item \emph{subword}: $\b_i$ is a marked subword of $\a_i$, for all $i \in \{1 \ldots k\}$;
\item \emph{nonemptiness}: $\b_1 \ldots \b_k \neq \varepsilon$, i.e., at least one of words $\b_i$ is nonempty; and
\item \emph{marking inheritance}: if $X$ is marked, say $X  \in \bar{\al{A}}_i$, then all symbols in $\b_i$ are marked.
\end{iteMize}
Observe that in the \emph{marking inheritance} condition we exploit possibility to mark additional occurrences
(in fact all occurrences, cf.~\eqref{eq:markmore}).
As an example, consider a transition:
\[
\mathrm{A} \ \trans \ \mathrm{AACABBCBCBC}, \ \mathrm{DDED}
\]
for $\mathrm{\{A,B,C\}}$ and $\mathrm{\{D,E\}}$ the alphabets of the two stacks. Then according to the 
colorings 
\[
\mathrm{A\mathbf{AC}A\mathbf{BB}C\mathbf{B}CBC} \qquad \text{ and } \qquad \mathrm{D\mathbf{D}\mathbf{E}D},
\] 
there are two marked subtransitions with $\mathrm{\bar{A}}$ the left-side symbol:
\[
\mathrm{\bar{A}} \ \trans \ \mathrm{\bar{A}\bar{A}\bar{C}\bar{C}\bar{B}\bar{C}} , \ \mathrm{\bar{D}D}
\qquad \text{ and } \qquad
\mathrm{\bar{A}} \ \trans \ \mathrm{\bar{A}\bar{A}\bar{C}\bar{C}\bar{B}\bar{C}} , \ \mathrm{\bar{D}\bar{D}}.
\]
Note that there are exponentially many different marked subtransitions, but each one is of polynomial size.
Finally, note that every subtransition is obtained from some transition by the marking procedure as above,
applied to every stack separately.

By the nonemptiness assumption on marked subtransitions we obtain a simple but crucial observation:
\begin{lem} \label{lem:nondecr}
Along a sequence of marked subtransitions, the size of configuration can not decrease.
\end{lem}

A \emph{marked subconfiguration} of a configuration $\confless{\a_1}{\ldots}{\a_k}$ is any tuple
$\confless{\b_1}{\ldots}{\b_k}$ such that $\b_i$ is a marked subword of $\a_i$ for all $i \in \{1 \ldots k\}$.

\begin{lem} \label{lem:subtrans}
For two configurations $\c{s}$ and $\c{t}$, 
the following conditions are equivalent:
\begin{enumerate}[\em(1)]
\item there is a sequence of transitions from $\c{s}$ to $\c{t}$,
\item there is a sequence of marked subtransitions from $\c{u}$ to $\c{t}$, for some marked subconfiguration $\c{u}$ of $\c{s}$. 
\end{enumerate}
\end{lem}
\begin{proof}
The implication from (1) to (2) follows immediately. The sequence of marked subtransitions is obtained by application
of the marking procedure to all transitions. 
For every transition, color in the marking procedure precisely those symbol occurrences that are irrelevant.

Now we show the implication from (2) to (1).  The proof uses \strnormedness. 

Assume a sequence $\pi$ of marked subtransitions from $\c{u}$ to $\c{t}$, for some marked subconfiguration $\c{u}$
of $\c{s}$. 
Recall that each subtransition in $\pi$ has its original transition of $\aut{A}$.
We claim that there is a sequence of transitions from $s$
to $t$, that contains the original transitions of all the marked subtransitions appearing in $\pi$, 
and \emph{canceling sequences} 
\begin{align} \label{eq:cancel-seq}
q \ \s{X} \ \trans \ \ldots \ \trans \ \config{q}{\eps}{\ldots}{\eps}
\end{align}
for some stack symbols $X$, existing due to \strnormedness\ assumption.

The sequence of transitions from $s$ to $t$ is constructed by reversing the marking procedure.
For the ease of presentation, beside letters from $\al{A}_i$, we will also use colored letters.

Start with the configuration $\c{s}$, and choose any coloring of symbol occurrences in $\c{s}$ that 
induces $\c{u}$ as the outcome of the marking procedure. 
% i.e., those symbol occurrences are colored that are not in $\c{u}$.
Then consecutively apply the following rule:
\begin{iteMize}{$\bullet$}
\item If the top-most symbol $\s{X}$ on some stack is colored, apply a canceling sequence for $\s{X}$.
\item Otherwise, apply  the original transition of the next subtransition from $\pi$, using again some coloring 
that could have been used in the marking procedure.
\end{iteMize}
For correctness, we need to show that all colored occurrences of symbols are eventually canceled out, 
as this guarantees that the final configuration is precisely $t$.

Let's inspect $\pi$. As no symbol in $t$ is marked, every marked symbol occurrence eventually disappears
as a result of firing of some subtransition. Recall that marking of a symbol $\bar{\s{X}}$ disappears only if
the subtransition pushes nothing on the stack of $\bar{X}$. 
As a consequence, every colored symbol occurrence will eventually appear on the top of its stack.
Thus the canceling sequence for $\s{X}$ will eventually be applied.
\end{proof}

\begin{lem} \label{lem:polynom}
For two configurations $\c{u}$ and $\c{v}$, 
if there is a sequence of marked subtransitions from $\c{u}$ to $\c{v}$, 
%for some configurations $\c{u}$ and $\c{v}$, 
then there is such a sequence of polynomial length wrt.~the sizes of $\c{u}$, $\c{v}$ and $\aut{A}$.
\end{lem}

\begin{proof}
From now on, we will write 'subtransitions' instead of 'marked subtransitions'.
As we will primarily work with subtransitions, we will use the stack alphabets $\al{A}_i \cup \al{\bar{A}}_i$ for $i \in \{1 \ldots k\}$.

The number of subtransitions that change state is bounded by the number of states of $\aut{A}$, 
as $\aut{A}$ is assumed to be \weak.
Thus it is sufficient to prove the lemma under the assumption that the subtransitions do not change state.
In other words, wlog.~we may assume $\aut{A}$ to be \stateless.
%\phcomm{Chyba nie tak prosto bo mogę mieć tranzycje które zmniejszają rozmiar konfiguracji bo zmieniaja stan, ale to nie wpływa}

The size of the right-hand side of a marked subtransition is at least $1$. Distinguish subtransitions
with the size of the right-hand side equal $1$, and call them
\emph{singleton subtransitions}. Clearly, the number of non-singleton subtransitions appearing in the sequence
in the above claim is at most equal to the size of $\c{v}$, thus it is sufficient to concentrate on the following claim:

\begin{clm} 
If there is a sequence of \emph{singleton} subtransitions from a configuration $\c{u}$ to $\c{v}$ then there is such a sequence of polynomial length.
\end{clm}

Note that the sizes of $\c{u}$ and $\c{v}$ in the above claim are necessarily the same.

Now we analyze in more detail the singleton subtransitions. Note that they have the form 
\begin{equation} \label{eq:singletonsubtrans}
\s{X} \trans \s{Y}
\end{equation}
as the right-hand sides contain precisely one occurrence of a symbol.
Consider the strongly connected components in the induced graph, with symbols as vertices, and singleton 
subtransitions~\eqref{eq:singletonsubtrans} as edges.

Distinguish those singleton subtransitions~\eqref{eq:singletonsubtrans} that stay inside one strongly connected component
(in other words, such that there is a sequence of subtransitions from $\s{Y}$ back to $\s{X}$)
and call them \emph{inner singleton subtransitions}.
Note that the number of non-inner subtransitions that appear in the sequence of the last claim is polynomial
(at most quadratic), thus the last claim is equivalent to the following one:

\begin{clm} 
If there is a sequence of \emph{inner singleton} subtransitions from a configuration $\c{u}$ to $\c{v}$ then there is such a sequence of polynomial length.
\end{clm}

The rest of the proof is devoted to showing the last claim.

We start by doing a sequence of simplifying assumption without losing generality.

First, wlog. we may assume that the relation~\eqref{eq:singletonsubtrans} is transitive, as we only care about the length
of the sequence of subtransitions up to a polynomial. 
Thus every strongly connected component is a directed clique.
%\phcomm{undirected clique, czemu directed??}

By the \emph{type} of a clique we mean the set of stacks that are represented in the clique, i.e.,
the stacks that have at least one symbol in the alphabet that belongs to the clique.
We may assume that there is no clique of singleton type. Indeed, otherwise the stack is essentially
inactive along the path, except for the top-most symbol, and thus may be ignored in our analysis.

Further, wlog. we may also assume that every clique has \emph{at most} one symbol belonging to every stack alphabet.
Indeed, two different symbols from the same clique and the same stack alphabet can easily mutate from one into the other,
when being the top-most symbol of the stack. And every symbol $\s{X}$ may be easily made top-most by popping all symbols
above $\s{X}$ to other stacks (this can be done due to the assumption that type of cliques are not singletons).

The simplifications lead us to the following reformulation of the last claim.
Let $k \geq 1$ be an integer.
Assume a finite set of symbols $\al{A}$, each symbol $\s{X} \in \al{A}$ coming with its type $\text{type}(\s{X}) \subseteq \{ 1 \ldots k \}$ of cardinality at least $2$. 
Consider the set of $k$-tuples of stacks $(\al{A}^*)^k$ 
satisfying the following consistency condition: if $\s{X}$ appears in the $i$th sequence then $i \in \text{type}(\s{X})$.
Consider the following transition rules: the top-most letter of some stack may be moved to the top of some other
stack, as long as the consistency is preserved.

\begin{clm}
If there is a sequence of transitions from some configuration $\c{u} \in (\al{A}^*)^k$ to some
configuration $\c{v} \in (\al{A}^*)^k$, then there is such a sequence of polynomial length.
\end{clm}
So formulated, the claim is not so hard.

We will show a polynomial sequence of transitions that starts in $\c{u}$ and ends in a configuration
$\c{u}'$ that has the same bottom-most symbol as $\c{v}$ on some stack. 
This is sufficient, as the same thing may be done for all other occurrences of symbols in $\c{v}$
(formally, an induction over stack depth is involved here).

Note that we do not assume that different symbols have different types.
Two symbols we call \emph{siblings} if they have the same type and this type has two elements
(thus the symbols may be placed only on two stacks).

Choose an arbitrary stack that is nonempty in $\c{v}$, say the $i$th stack, with the bottom-most symbol $\s{X}$.
We may assume wlog. that $\s{X}$ does not appear in $\c{u}$ on the $i$th stack 
(otherwise, i.e., if some occurrences of $\s{X}$ in $\c{u}$ are on the $i$th stack,
move all the occurrences of $\s{X}$, together with all other symbols above them, to arbitrary other stacks).

Let the $j$th stack in $\c{u}$ contain an occurrence of symbol $\s{X}$, for some $j \neq i$.

The sequence of steps from $\c{u}$ to $\c{u}'$ is the following:
\begin{enumerate}[(1)]
\item Move all symbols above the chosen occurrence of $\s{X}$ from the $j$th stack to other stacks.  % different than the $i$th one. 
\item Move all symbols from the $i$th stack to other stacks such that $\s{X}$ is still on the top of the $j$th stack.
%\item Move all symbols from the $i$th stack to other stacks different than the $j$th one.
\item Move the chosen occurrence of symbol $\s{X}$ to the destination $i$th stack.
\end{enumerate}
Clearly step 1.~ may be always done.
We will show that step 2.~may be always done as well. We distinguish two cases.

If the symbol $\s{X}$ is not a sibling, every other symbol may be moved, from the $i$th stack,
to a stack different than the $i$th one, in such a way that after this operation $X$ will be still on the top of the $j$th stack.
Indeed, assume that a symbol $\s{Y}$ is on the top of the $i$th stack. If 
$\s{Y}$ can be moved to a stack different than the $j$th one, we are done.
Otherwise, $\s{Y}$ can only occur on the $i$th and $j$th stacks. 
According to the assumption, $\s{X}$ and $\s{Y}$ are not siblings, thus there is another $k$th stack  
to which $\s{X}$ can be moved. Then we proceed as follows: $\s{X}$ is moved from $j$th to the $k$th stack,
next $\s{Y}$ is moved from the $i$th stack to the $j$th stack, and finally $\s{X}$ is moved back from the $k$th stack to the
$j$th stack.

As the second case, assume that $\s{X}$ is a sibling. Assume that the top-most occurrence of $\s{X}$ 
on the $j$th stack has been chosen. As $j \neq i$, and there is a sequence of steps
from $\c{u}$ to $\c{v}$, one easily observes that no sibling of $\s{X}$ may occur in $\c{u}$ either 
on the $i$th stack, or above $\s{X}$ on the $j$th stack. Thus step 2.~clearly can be done.

This completes the proof of Lemma~\ref{lem:polynom} and thus also the proof of Lemma~\ref{lem:np1-proof-step1},
under the stateless assumption.
%\qed
\end{proof}

\smallsection{Decision procedure.}
Now we drop the stateless assumption.
Note that the notion of marked subconfiguration and marked subtransition may be easily adapted to transitions that change state.
We do not impose however the nonemptiness condition on transitions that change state, which is in accordance with the intuition that
irrelevant symbol occurrences are removed in the marking procedure. 
Using Lemmas~\ref{lem:nondecr}, ~\ref{lem:subtrans} and~\ref{lem:polynom} we will define 
the nondeterministic decision procedure for \strnormed\ \weak\ \MPDAs.

Let the two given configurations $\c{s}$ and $\c{t}$ have control states $q$ and $p$, respectively.
In the first step, the algorithm guesses a number of marked subconfigurations $\c{t}_1 \ldots \c{t}_{n-1}$,
where $n$ is not greater than the number of states of $\aut{A}$, and
marked subtransitions that change state:
\begin{align*}
\c{t_1} & \trans \c{s_1} & \c{t_2} & \trans \c{s_2} & & \ldots & \c{t_{n-1}} \trans \c{s}_{n-1}
\end{align*}
such that $\c{s}_i$ and $\c{t}_{i+1}$ have the same  control states for $i \in \{0 \ldots n-1\}$.
Note that size of all configurations $\c{s}_i$ and $\c{t}_i$ is necessarily bounded by
$\text{size}(\c{t}) + |\aut{A}|$, as the size of a configuration along the run may be decreased
only during the change of state and even then only by one.
For convenience, we write $\c{s}_0$ instead of $\c{s}$ and $\c{t}_n$ instead of $\c{t}$.
In particular, we assume that the control state of $\c{t_1}$ is $q$, and the control state of $\c{s}_{n-1}$ is $p$.
Relying on Lemma~\ref{lem:nondecr}, it is sufficient to consider configurations of sizes satisfying the following inequalities:
\begin{align}\label{eq:rel}
\text{size}(\c{s}_i) \leq \text{size}(\c{t}_{i+1})   \qquad \text{ for } i \in \{1 \ldots n-1\} .
\end{align}

In the second phase, the algorithm guesses, for $i \in \{0 \ldots n-1\}$,
% a subconfiguration $\bar{\c{s}}_i$ of $\c{s}_i$, and 
a sequence of subtransitions from $\c{s}_i$ to $\c{t}_{i+1}$
of length bounded by polynomial derived from the proof of Lemma~\ref{lem:polynom}; and
checks that the respective sequences of subtransitions lead from $\c{s}_i$ to $\c{t}_{i+1}$, 
as required by Lemma~\ref{lem:subtrans}.
\end{proofof}

\begin{proofof}{Lemma~\ref{lem:np1-proof-step2}}
Suppose $\aut{A}$ is a \strnormed\ \weak\ \MPDA. Let $L$, $K$ be regular sets of configurations of $\aut{A}$
and let $\pi$ be a sequence of transitions
from some configuration $\c{s} \in L$ to some configuration $\c{t} \in K$.
We will demonstrate existence of configurations $\c{s}' \in L$ and $\c{t}' \in K$ such that 
$\c{t}'$ is of polynomial size
%\phcomm{ja bym precyzyjnie napisał wilomian od czego,bo dopiero w następnym zadaniu piszemy, że nie od $s$ i to nie wprost} 
and $\c{s}' \reaches \c{t}'$.
Importantly, we do not have to provide any bound on the size of $\c{s}'$, as the polynomial bound  
follows by Lemma~\ref{lem:reduceproblems}.

Recall the coloring discipline used in the proof of Lemma~\ref{lem:np1-proof-step1}.
There we used just one color; here we will use an unbounded number of different colors,
as described below.

The coloring discipline will apply to all configurations appearing in $\pi$.
We start by coloring all symbol occurrences in the first configuration $\c{s}$ with different colors.
% The colors appearing in $\c{s}$ we call \emph{root colors}.
Then, for every transition $\c{s}_1 \trans \c{s}_2$ in $\pi$, assumed that $\c{s}_1$ has already been colored,
we stipulate the following coloring rule for $\c{s}_2$ (recall that symbol occurrences in $\c{s}_2$ are divided
into those corresponding to symbol occurrences in $\c{s}_1$, and fresh ones):
\begin{iteMize}{$\bullet$}
\item If a symbol occurrence corresponds to a symbol occurrence in $\c{s}_1$, its color is the same as the color of corresponding symbol occurrence.
\item Let $c$ be the color of the unique occurrence of symbol in $\c{s}_1$, say symbol $\s{X}$, that is involved in the transition.
All fresh symbol occurrences in $\c{s}_2$ that appear on the same stack as $\s{X}$ are colored with $c$; 
we say that they \emph{inherit} the color from $\s{X}$.
All other fresh symbol occurrences in $\c{s}_2$ are colored by new fresh colors: one for each other stack.
% , with the proviso that two occurrences have the same color if and only if they occur on the same stack. 
Thus at most $k-1$ new fresh colors is needed for coloring fresh occurrences on other stacks, where $k$ is the number of stacks.
\end{iteMize}
For any color used, and for any fixed configuration, the set of all symbol occurrences colored with that color we call \emph{line}.
Note that a line is always a subset of symbol occurrences on a single stack.
However, there are typically many lines sharing their color, but this happens only to lines from different configurations. 
Further, note that the cardinality of a line is not bounded in principle, due to the inheritance of color.

We now aim at reducing the size of the destination configuration $\c{t}'$.
Roughly speaking, we will prove that $\c{s}' \reaches \c{t}'$, for some $\c{s}' \in L$ and $\c{t}' \in K$ such that both the number of
different lines in $\c{t}'$, and the cardinality of all lines in $\c{t}'$, are polynomially bounded.

For convenience, we split colors into two disjoint subsets.
A color $c$ is called \emph{active} if some symbol occurrence labeled by $c$:
\begin{iteMize}{$\bullet$}
\item either is involved in some transition in $\pi$,
\item or is a fresh symbol occurrence created by some transition in $\pi$.
\end{iteMize}
Otherwise, a color is called \emph{inactive}, i.e., occurrences of this color are present in $\c{s}$
and are not involved in any transition during the run.
Likewise, a line is also called active or inactive, according to its color.
Note that  inactive colors label suffixes of stacks in every configuration in $\pi$, and these suffixes do not change along $\pi$. Inactive lines are clearly singletons.

\smallsection{Bounding the number of active lines.}
Consider content of some stack, say the $i$th stack, in the destination configuration $\c{t} \in K$.
Denote by $w \in \al{A}_i^*$ its prefix colored by active colors.
Every active line on the $i$th stack corresponds to an infix of $w$, and thus the coloring induces a factorization 
\[
w = w_1 \cdot w_2 \cdot \ldots \cdot w_m 
\]
determined by some $m-1$ positions in $w$.

Recall that the language $K$ is represented by a tuple $\aut{B}_1 \ldots \aut{B}_k$ of finite automata, 
one automaton per stack. 
A run of the automaton $\aut{B}_i$ over the $i$th stack of $\c{t}$ labels each of the $m-1$ distinguished 
positions in $w$ by a state. 
By a standard pumping argument, there is a subword $w'$ of $w$, obtained by removing a number of lines from $w$,
that contains at most as many lines 
as the number of states of $\aut{B}_i$, and such that $\aut{B}_i$ reaches the same state after reading $w$ and $w'$.
By repeating the pumping argument for all stacks, one obtains a configuration $\c{t}'$ still belonging to $K$,
that contains only a polynomial (in fact, at most quadratic) number of active lines, as required.

We only need to show that $\c{s} \reaches \c{t}'$.
In this part of the proof we will use the canceling sequences~\eqref{eq:cancel-seq}, available due to \strnormedness.
Observe that  every active line that appears in $\pi$ appears as the top-most one on its stack at some configuration in $\pi$.
We apply the canceling sequence for all symbol occurrences in every active line not appearing in $\c{t}'$.
In order to keep the reachability $\c{s} \reaches \c{t}'$, we apply the canceling sequence
in the last configuration in $\pi$ where this line is the top-most one.
Thus the disappearance of a line has no effect for the remaining lines.

\smallsection{Bounding the number of inactive lines.}
We repeat a pumping argument similar to the above one.
Let $L$ and $K$ be represented by $\aut{A}_1 \ldots \aut{A}_k$ and  $\aut{B}_1 \ldots \aut{B}_k$, respectively. 
Consider some $i \leq k$ and runs of automata $\aut{A}_i$ and $\aut{B}_i$ over the inactive suffix of the $i$th stack of $\c{t}'$
(or $\c{t}$). The runs label every position by a pair of states of $\aut{A}_i$ and $\aut{B}_i$, respectively.
Upon a repetition of the same pair of states, a standard pumping applies.
Thus the length of every inactive suffix in $\c{v}$ may be reduced to at most quadratic.

\smallsection{Bounding the cardinalities of active lines.}
Consider the configuration $\c{t}'$ obtained by now, and
a single active line on some $i$th stack in this configuration. Let $v \in \al{A}_i^*$ be the word representing the line.
Thus the $i$th stack in $\c{t}'$ is of the form:
\[
w_1 v \, w_2
\]
for some words $w_1, w_2$.
Similarly as before, we aim at applying pumping inside $v$, to reduce its length.

Let's focus on symbol occurrences in $v$ in configuration $\c{t}'$ 
and on the corresponding symbol occurrences in other configurations in $\pi$.
Observe that all symbol occurrences in $v$ satisfy the following condition:
\begin{quote}
\emph{the corresponding symbol occurrence in some previous configuration was freshly created in some transition}.
\end{quote}
%Distinguish those symbol occurrences in $v$ that satisfy the following condition:
%\begin{quote}
%\emph{the corresponding symbol occurrence in some previous configuration was involved in a transition}.
%\end{quote}
Some of the above-mentioned transitions have created new lines, and some not.
%Among those distinguished
Among symbols in $v$, distinguish a subset containing only those occurrences that satisfy the 
following strengthened condition:
\begin{quote}
\emph{the corresponding symbol occurrence in some previous configuration was freshly created in some transition
that created a new line that is represented in $\c{t}'$}.
\end{quote}
The distinguished symbol occurrences call \emph{non-local}, the others call \emph{local}.

The overall number of lines in $\c{t}'$ is polynomially bounded, thus the same bound applies to the number of
non-local symbol occurrences in $\c{v}$. We thus obtain:
\begin{clm}
There is only polynomially many non-local symbol occurrences in $v$.
\end{clm}
Thus, it is sufficient to reduce the length of any block of local occurrences in $v$.
From now on we focus on a single maximal infix $u$ of $v$ that contains only local symbol occurrences.

Those transitions in $\pi$ that involve a symbol occurrence corresponding to a symbol occurrence in $u$
use only the $i$th stack. Thus this set of transitions is essentially a \stateless\ pushdown automaton.
We will use a well-known fact (see for instance~\cite{BEM97}):
\begin{clm}%[\cite{C92}]
For a pushdown automaton one can construct a finite automaton of polynomial size that recognizes
the language of all reachable configurations of the pushdown automaton.
\end{clm}
We will now use a standard pumping to reduce the length of $u$.
As said above, a run of $\aut{B}_i$ over the $i$th stack of $\c{v}$ labels each position of $u$ with a state.
Likewise for the automaton of the above claim.
Upon a repetition of a pair of states, a standard pumping applies, as usual.
This completes the proof of Lemma~\ref{lem:np1-proof-step2}.
\end{proofof}

%\subsection{Proof of Theorem~\ref{thm:np2}}
%A straightforward adaptation of the proof of Lemma~\ref{lem:np1-proof-step1} (combined with Lemma~\ref{lem:reduceproblems}).
%Observe that irrelevant symbol occurrences must necessarily be normed, as they do not contribute
%to the target configuration.

%%% Local Variables: 
%%% TeX-master: "main"
%%% End: 

\section{Decidability}
\label{sec:decid}

In this section we provide proofs of Theorems~\ref{thm:decid}, \ref{thm:reach-set} and~\ref{thm:decid_mpda}.

\bigskip

\restate{Theorem~\ref{thm:decid}}{\thmdecid}

\begin{proof}
By virtue of Lemma~\ref{lem:reduceproblems} we may focus on the \singsing\ reachability only.
Fix an \MPDA\ $\aut{A}$ and two configurations $\c{s}$ and $\c{t}$. We will describe
an algorithm to decide whether $\c{s} \leadsto_\aut{A} \c{t}$.
Roughly speaking, our approach is to define a suitable well order compatible with transitions, and then apply
a standard algorithm for reachability of a downward-closure of a configuration\footnote{The algorithm works actually in 
any well-structured transition system under suitable assumptions, see~\cite{FS01}. Theorem~\ref{thm:decid} is in fact a special
case of Theorem 5.5 in~\cite{FS01}.}.
However, to apply the standard framework we need to introduce some additional structure in configurations.
This additional structure will be intuitively described as coloring of symbols, similarly as marking in the proof of 
Lemma~\ref{lem:np1-proof-step1}.

Recall the notion of relevant symbol occurrences, introduced in Section~\ref{sec:proofoflemma}.
The idea of the proof will be based on the observation that removing some irrelevant
symbol occurrences has no impact on reachability of a fixed target configuration
(cf. Lemma~\ref{lem:order} from Section~\ref{sec:proofoflemma}).

Fix the target configuration $\c{t}$. We will define \emph{colored configurations} and modified transitions 
between colored configurations.
% , in order to take into account reachability of the configuration $\c{t}$. 
The basic intuition is that irrelevant symbol occurrences will be colored.
Note however that we don't know in advance which symbol occurrences in a given configuration $\c{s}$
are relevant and which are not, as we do not even know if $\c{s} \leadsto \c{t}$.
Thus a coloring will have to be guessed.

Let $n$ be the number of states of $\aut{A}$ and let $m$ be the size of $\c{t}$.
By a \emph{colored configuration} we mean a configuration with some symbol occurrences colored,
such that the number of uncolored symbol occurrences is smaller than $n + m$.
Formally, coloring is implemented by extending the alphabet of every stack with its colored copy.
We define an ordering on colored configurations:
$\c{r}' \preceq \c{r}$ if 
$\c{r}'$ is obtained from $\c{r}$ by removing some colored symbol occurrences.
(In particular, if $\c{r}' \preceq \c{r}$ then both configurations are identical, when restricted to uncolored symbols.
Thus an uncolored configuration corresponds to a downward-closed set.) 
As the number of uncolored occurrences is bounded, the number of blocks of colored occurrences is bounded likewise.
The ordering $\preceq$ is like a Higman ordering on words (i.e., a word $v$ is smaller than a word $w$ if
$v$ can be obtained from $w$ by removing some letters), extended in the point-wise manner to blocks of colored occurrences. 
Thus one easily shows, using Higman's lemma:
\begin{clm}
The ordering $\preceq$ is a well order on colored configurations, i.e. for any infinite sequence of configurations
$\c{s}_1 \c{s}_2 \ldots$ there are indices $i < j$ such that $\c{s}_i \preceq \c{s}_j$.
\end{clm}

Now we define the transition rules for colored configurations.
Consider any original transition rule $\delta$ of $\aut{A}$. This transition rule will give rise to a number of 
new transition rules that will be applicable to colored configurations.
One new transition is obtained by coloring all symbols in $\delta$, i.e., both the left-hand side symbol 
and all the right-hand side symbol occurrences.
In all other new transitions arising from $\delta$, the left-hand side symbol is kept uncolored.
On the other hand, an arbitrary subset of the right-hand side symbol occurrences may be colored,
under the following restriction:
\begin{quote}
if the transition $\delta$ does not change state then at least one of right-hand side symbols must be kept uncolored.
\end{quote}
This corresponds to the intuition that uncolored symbol occurrences correspond to relevant ones.

We have thus now two transition systems: the original transition system and the colored one.
The relationship between reachability in these two systems is stated in the following claim 
(recall that the configuration $\c{t}$ is fixed and contains no colored symbols):
\begin{clm}
For any configuration $\c{s}$, $\c{s} \reaches \c{t}$ if and only if
there is some coloring $\c{s}'$ of $\c{s}$ such that $\c{s}' \reaches \c{t}$.
\end{clm}
Indeed, the only if direction is obtained by coloring precisely irrelevant symbol occurrences in $\c{s}$.
The if direction also follows immediately, by replacing the colored transitions with their
uncolored original transitions.

Basing on the above claim, the algorithm for $\c{s} \reaches \c{t}$ simply guesses a coloring
$\c{s}'$ of $\c{s}$ and then checks if $\c{s}' \reaches \c{t}$ in the colored transition system.
It thus only remains to show that the latter problem is decidable.
For this we will need a compatibility property of the colored transitions with 
respect to the well order: 
\begin{clm}
For every colored configurations $\c{r}', \c{r}$ and $\c{u}$,
if $\c{r}' \preceq \c{r} \trans \c{u}$ then 
\begin{iteMize}{$\bullet$}
\item either there is a colored configuration $\c{u}'$ with $\c{r}' \trans \c{u}' \preceq \c{u}$,
\item or $\c{r}' \preceq \c{u}$.
\end{iteMize}
\end{clm}
In other words, $\preceq$ is a variant of backward simulation with respect to $\trans$.
Indeed, if the symbol occurrence involved in $\c{r} \trans \c{u}$ is uncolored,
the transition may be also fired from $\c{r}'$.
Otherwise, suppose that the symbol occurrence involved in $\c{r} \trans \c{u}$ is colored
(recall that in this case all fresh symbol occurrences are colored). 
If this occurrence appears also in $\c{r}'$, it may be fired similarly as above.
On the other hand, if this occurrence does not appear in $\c{r}'$, we have $\c{r'} \preceq \c{u}$,
as required.

Using the last claim we easily show decidability.
The algorithm explores exhaustively a portion of the tree of colored configurations reachable from $\c{s}'$,
with the following termination condition.
As the ordering $\preceq$ is a well order, we know that on every path eventually
two colored configurations appear, say $\c{u}'$ and $\c{u}$, such that $\c{u}'$ precedes $\c{u}$ and
$\c{u}' \preceq \c{u}$. Such a pair we call \emph{domination pair}.
Whenever a domination pair is found on some path, the algorithm stops lengthening this path.
The well order guarantees thus that the algorithm terminates, after computing a finite tree of colored configurations
(finiteness of the tree follows by K{\"o}nig's lemma).
The algorithm answers 'yes' if the configuration $\c{t}$ appears in the tree.

Now we prove correctness of the algorithm.
Towards contradiction, suppose $\c{t}$ is reachable from $\c{s}'$ but $\c{t}$ is not found in the tree.
Consider the shortest path $\pi$ from $\c{s}'$ to $\c{t}$, and the domination pair $\c{u}' \preceq \c{u}$ on that path.
Thus $\c{u} \reaches \c{t}$.
Using the compatibility condition, we deduce that $\c{u} \reaches \c{t}$ implies 
$\c{u}' \reaches \c{t}'$ for some $\c{t}' \preceq \c{t}$, along a path not longer that the path from $\c{u}$ to $\c{t}$. 
By the definition of $\preceq$ we obtain $\c{t}' = \c{t}$.
Thus the fragment of path $\pi$ from $\c{s}'$ to $\c{u}'$, composed with the path from $\c{u}'$ to $\c{t}$ yields
a path strictly shorter than $\pi$, a contradiction.
\end{proof}

\bigskip

\restate{Theorem~\ref{thm:reach-set}}{\thmreachset}

\begin{proof}
Consider a \strnormed\ \MPDA, and a regular set $L$ of configurations.
Denote by $B$ the backward reachability set of $L$. We aim at proving that $B$ is regular.
We do not claim however that a representation of this regular set may be effectively computed in general.

The proof is in two steps. In the first one we introduce and justify a simplification of the problem.
As the second step, we prove regularity of the backward reachability in the simplified case.

\smallsection{Restriction to fully active paths.}
%For control states $q$, 
Let $L_q$ (resp.~$B_q$) denote restrictions of $L$ (resp.~$B$) 
to configurations with control state $q$.
It is enough to prove regularity of the set
\[
B_{p,q} \ = \ \{ \c{s} \in B_p : \c{s} \reaches  L_q \} 
\]
for every pair of states $p, q$, due to the following equalities:
\[
B_p = \bigcup_q B_{p,q} \qquad \qquad B = \bigcup_p B_p 
\]
and closure of regular languages under finite unions.
Fix two control states $p$ and $q$ from now on. 
Denote by $\al{S}_q$ the set of all configurations with the control state $q$.

%As all the configurations from $B_p$ (resp.~$L_q$) have the same fixed control state, we will treat them as stateless configurations whenever convenient.
% That is, we may think of $B_p$ (resp.~$L_q$) as a  subset of
%$\al{S} = \al{S}_1^* \times \ldots \times \al{S}_k^*$, the set of all tuples over respective stack alphabets,
%$k$ being the number of stacks.

%\begin{align*} \label{eq:split}
%L_q & = \{ \confless{\a_1}{\ldots}{\a_k} : \config{q}{\a_1}{\ldots}{\a_k} \in L \} \\
%B_q & = \{ \confless{\a_1}{\ldots}{\a_k} : \config{q}{\a_1}{\ldots}{\a_k} \in B \},
%\end{align*}

Recall that the set $L_q$ is recognized by some monoid homomorphism, i.e.,
there is $h : \al{S}_q \to M$ for some finite monoid $M$ and for every $q$, $L_q = h^{-1}(N_q)$ for some subset $N_q \subseteq M$. Note that, formally speaking, the state $q$ is a part of every configuration of $L_q$, hence
the domain of the homomorphism are all configurations with state $q$.
Thus formally speaking the homomorphism respects the multiplication of configurations defined as:
\[
\config{q}{\a_1}{\ldots}{\a_k} \cdot \config{q}{\b_1}{\ldots}{\b_k} \ = \
\config{q}{\a_1 \b_1}{\ldots}{\a_k \b_k},
\]
for $k$ the number of stacks.
As $L_q$ is recognized by $h$, we have
\begin{align} \label{eq:rownanie}
L_q = \bigcup_{m, n} h^{-1}(m) \,  h^{-1}(n),
\end{align}
where $m$ and $n$ range over pairs of elements of $M$ satisfying $m \cdot n \in N_q$.

A path from a configuration $\c{s}$ to some target configuration is \emph{fully active} if some descendant of every symbol occurrence
in $\c{s}$ is involved in some transition.
For $m \in M$, denote by $F_{m}$ the set of all configurations $\c{s}$ with control state $p$ such that
there is a fully active path from $\c{s}$ to some configuration $\c{t}$ in $h^{-1}(m)$
(the control state of $\c{t}$ is necessarily $q$, by the definition of $h$).
We claim that the set $B_{p,q}$ is equal to the following union:
\begin{align*} \label{eq:split}
B_{p,q} \ = \ \bigcup_{m, n} F_{m} \, \, h^{-1}(n),
\end{align*}
with $m, n$ ranging over all pairs $m, n \in M$ such that $m \cdot n \in N_q$, similarly as in~\eqref{eq:rownanie}.

Indeed, if a configuration $\c{s}$ belongs to $ F_{m} \, \, h^{-1}(n)$ for some $m \cdot n \in N_q$, then 
%by a fully active path 
it can reach a configuration in $ h^{-1}(m) \,  h^{-1}(n) \subseteq L_q$.
On the other side, if some configuration $\c{s}$ with control state $p$ 
can reach a configuration $\c{t} \in L_q$, we
can split $\c{s}$ and $\c{t}$ into two parts 
\begin{align*}
\c{s} = \c{s}_1 \cdot \c{u} \qquad \c{t} = \c{t}_1 \cdot \c{u}
\end{align*}
such that
\begin{iteMize}{$\bullet$}
  \item considering the transitions on the path $\c{s} \reaches \c{t}$, they form a fully active path
from $\c{s}_1$ to $\c{t}_1$,
  \item the symbol occurrences appearing in $\c{u}$ do not make any transition.
\end{iteMize}
Therefore, for $m = h(\c{t}_1)$ we obtain $\c{s_1} \in F_m$.
In consequence $\c{s} \in F_m \, \, h^{-1}(n)$ for some $m, n \in M$ such that $m \cdot n \in N_q$, as required.

We have thus demonstrated that we can safely restrict to fully active paths.

\smallsection{The proof for fully active paths.}
Under the restriction to fully active paths, the proof is fairly easy.
Let $\aut{A}$ be an \MPDA\ and let $L$ be a regular set of configurations.
% Recall that we may assume an \MPDA\ to be \stateless\ and \strnormed.

The \emph{bottom-fixed Higman ordering} over words relates $w'$ and $w$
if the words have the same last letter (recall that the last letter corresponds to the bottom element of a stack), and the remaining prefixes of $w'$ and $w$ are related by the ordinary Higman ordering.
Order the configurations as follows: the order, denoted by $\preceq$, relates two configurations if
\begin{iteMize}{$\bullet$}
\item the control states are the same, and
\item the stack contents are related by the point-wise extension of the bottom-fixed Higman ordering.
\end{iteMize}
Observe that this order  is a well order.

\begin{clm}\label{cl:fully_active_closed}
Under the restriction to fully active paths, the backward reachability set of a regular set $L$
is upward closed with respect to $\preceq$.
\end{clm}
Indeed, assuming $\c{s'} \preceq \c{s}$ and $\c{s'} \leadsto \c{t} \in L$, one deduces
$\c{s} \leadsto \c{t}$ by applying the canceling sequences.
These sequences are applicable to some descendant of every symbol occurrence in $\c{s}$, as the path is fully active. (It is here where  the \strnormedness\ assumption is needed.)
Note the reason why we have chosen the bottom-fixed variant of Higman ordering:
this prevents $\c{s}$ to have some additional symbols below some bottom symbol of $\c{s'}$.

By the above claim, the backward reachability set is determined by the minimal configurations with respect to $\preceq$.
As $\preceq$ is a well order % By the fact that order is $WQO$\todo{Not because of Higman lemma but because it is WQO}
 there is only finitely many minimal configurations, and thus the backward reachability set is regular.
%
% \smallsection{Effectively computable backward reachability sets}
% Assume additionally an \MPDA\ to be \weak.
% For effectivity, we inspect the proofs of Lemmas~\ref{lem:reduceproblems} and~\ref{lem:np1-proof-step2} and conclude
% that all the minimal elements are of polynomial size with respect to the size of $L$:
% indeed, if $\c{s} \leadsto L$ then $\c{s}' \leadsto L$ for some $\c{s}' \preceq \c{s}$ of polynomial size. It is important to notice that
% the path remains fully active while decreasing the size of the source configuration.
% %, in the proof of Lemma~\ref{lem:reduceproblems}.
% With this, the algorithm determines the minimal elements by inspecting exhaustively all configurations 
% $\c{s}$ of polynomially bounded size, and checks
% for every of them if $\c{s} \leadsto L$. The restriction to only fully active paths
% may be imposed by a simple encoding.
%
% The procedure may be implemented in exponential time.
\end{proof}

\bigskip

\restate{Theorem~\ref{thm:decid_mpda}}{\thmdecidmpda}

\begin{proof}
We are going to use Theorem~\ref{thm:reach-set}.
Consider two regular sets $L, K$ of configurations. The decision procedure for $L \reaches K$ runs two
partial procedures. The positive semi-procedure searches for a witness for reachability, by an exhaustive enumeration of all finite paths starting in some configuration of $L$ and ending in some configuration of $K$.
The core of the difficulty lies in the negative semi-procedure, to be described below.

The negative semi-procedure searches for a \emph{separator}.
Knowing that the backward reachability set of $K$ is regular, we define separator
as any regular set $M$ of configurations satisfying the following conditions:
\begin{iteMize}{$\bullet$}
\item $K \subseteq M$,
\item $L \cap M = \emptyset$,
\item $M$ is backward closed: whenever $\c{s} \trans M$ then $\c{s} \in M$.
\end{iteMize}
Clearly, if $L \not\reaches K$ then the backward reachability set of $K$ satisfies all the three conditions above.
By Theorem~\ref{thm:reach-set} we know that the backward reachability set is regular, and thus 
\[
L \not\reaches K \qquad \text{ iff } \qquad \text{there is some separator } M.
\]
The negative semi-procedure enumerates all regular sets $M$, and for every of them checks the three conditions above. The first two are easily shown decidable as boolean operations on regular languages yield regular languages and are effectively computable.
Decidability of the third condition follows from an observation that if $M$ is regular than
the predecessor set $\{ \c{s} : \c{s} \trans M \}$ is also regular, and may be effectively constructed.
\end{proof}

%%% Local Variables: 
%%% TeX-master: "main"
%%% End: 

\section{Undecidability}
\label{sec:undecid}

Here we provide proofs of Theorems~\ref{thm:undecid} and~\ref{thm:undecid-relaxed}.

\bigskip 

\restate{Theorem~\ref{thm:undecid}}{\thmundecid}

\begin{proof}
We start by considering \stateless\ \unnormed\ MPDAs.
We reduce the problem of checking if the intersection of two context-free languages is empty.

Assume two context-free grammars in Greibach Normal Form over an input alphabet $A$. 
We will construct an \MPDA\ with three stacks.   % \slcomm{a dla dwoch stosow?} 
Two stacks will be used to simulate derivations of the two grammars, 
and the other third stack will be used for storage of the input word. Formally,
the alphabet of the first and second stack are the nonterminals of the two grammars, and
the alphabet of the third stack contains two symbols $a_1$ and $a_2$ for every terminal symbol $a$ of the grammars.
For every production 
\begin{equation} \label{eq:prod}
\s{X} \to a \, \alpha
\end{equation}
of the first grammar, there will be a transition
\[
\s{X} \ \trans \ \a, \, \eps, \, a_1
\]
that drops $\a$ on the first stack and $a$ on the third one. Likewise, for every production~\eqref{eq:prod}
of the second grammar, there is a transition
\[
\s{X} \ \trans \ \eps, \, \a, \, a_2 .
\]
The initial configuration is $\confless{\s{X}_1}{\s{X}_2}{\eps}$, where $\s{X}_i$ is the initial symbol of the
$i$th grammar.
Finally, the regular language $L$ of target configurations constraints 
the first two stacks to be empty, and the third one to:
\[
\{ a_1 a_2 : a \in A \}^* .
\]
One easily verifies that the intersection of the two grammars is nonempty if and only if
some configuration from $L$ is reachable from the initial configuration.

Now we turn to \weak\ \normed\ \MPDAs.
It turns out that \normedness\ assumption does not make reachability problem easier, in case of \weak\ automata.
Indeed, the case of  \stateless\ \unnormed\ \MPDAs\ easily reduces to the case of \weak\ \normed\ \MPDAs.
It is sufficient to add an additional sink state, and for every symbol $\s{X}$ two additional transitions, to enforce \normedness.
The first one allows $\s{X}$ to change state to the sink state. The other one allows $\s{X}$ to disappear
in the sink state.
(This is in fact a reduction of the whole case of \weak\ \unnormed\ \MPDAs.)
\end{proof}

\bigskip

\restate{Theorem~\ref{thm:undecid-relaxed}}{\thmundecidrelaxed}

\begin{proof}
The proof is    % We show Theorem~\ref{thm:undecid-relaxed} 
by reduction from the Post Correspondence Problem (PCP).
For a given instance of PCP, consisting of a finite set of pairs $(s_i, t_i)$ of words,
$i \in \{1 \ldots n\}$, we construct a \stateless\ \strnormed\ \MPDA\ $\aut{A}$
and a relaxed-regular set $L$ such that the PCP instance has a solution 
\[
s_{i_1} \, s_{i_2} \ldots s_{i_k} \ \ = \ \  t_{i_1} \, t_{i_2} \ldots t_{i_k} \qquad
(i_j \in \{1 \ldots n\} \text{ for } j \in \{1 \ldots k\})
\]
if and only if there exists a path from the initial configuration of $\aut{A}$ to $L$.
Roughly speaking, a run of $\aut{A}$ will simply guess a PCP solution, and the target language $L$ will be used
to check its correctness. 

The main difficulty to overcome is the \strnormedness\ requirement, which implies that every symbol
may always disappear and not contribute to the target configuration.

\smallsection{Half-solution.}
We start by restricting to only the left-hand side words $s_i$ of the PCP instance.
We will construct an \MPDA\ $\aut{A}_1$,  and a relaxed-regular language $L_1$ of configurations,
so that the reachable configurations of $\aut{A}_1$ belonging to $L_1$ are essentially of the form (two stacks):
\begin{align} \label{eq:idea}
(i_1 \, i_2 \ldots i_k, \ \ s_{i_1} \, s_{i_{2}} \ldots s_{i_k}).
\end{align}
In other words, one of stacks contains the sequence of indexes, and the other one contains the concatenation of the corresponding
words $s_i$.

For technical reasons we will however need four stacks and few auxiliary nonterminal symbols.
%Later we combine two such half-solutions.
The nonterminals of $\aut{A}_1$ are following (superscripts indicate the stack number of every nonterminal):
\begin{iteMize}{$\bullet$}
  \item $\s{G}^1$  %, \s{G'}^1$ 
  and $\s{G}^4$,  %, \s{G'}^4$,    
  used for 'guarding' symbols on their stacks, as described below;
  \item $i^1$ and $i^2$, for $i \in\{1 \ldots n\}$,   representing the $i$th word $s_i$; 
  \item $a^3$ and $a^4$, for $a \in \Sigma$, representing alphabet letters of the PCP instance.
\end{iteMize}

The initial configuration is $(\s{G}^1, \ \eps, \ \eps, \ \s{G}^4)$. 

For a word $w = a_1 a_2 \ldots a_{m} \in \Sigma^*$, we write
$w^3$ to mean the word $a_1^3 a_2^3 \ldots a_{m}^3$.
Likewise for $w^4$.
The transition rules of $\aut{A}_1$ are the following. For $i \in \{1 \ldots n\}$, there are rules:
\begin{align*}
\s{G}^1 \ & \trans \ \s{G}^1 \, i^1, \ \eps, \ s_i^3, \  \eps  &
\s{G}^4 \ & \trans \ \eps, \ i^2, \ \eps, \ \ \s{G}^4 \, s_i^4 .
\end{align*}
Additionally, to fulfill the \strnormedness\ restriction we add \emph{disappearing transition rules} of the form 
$
X \ \ \trans \ \ \eps, \ \eps, \ \eps, \ \eps
$
for all nonterminal symbols. 

The target set $L_1$ is defined to contain all configurations of the form
\[
(\s{G}^1 \, \alpha_1, \ \alpha_2, \ \alpha_3, \ \s{G}^4 \, \alpha_4),
\]
with $\alpha_1$ almost equal to $\alpha_2$ and $\alpha_3$ almost equal to $\alpha_4$.
By 'almost equal' we mean equality modulo (ignoring) the superscripts.
The set $L_1$ is clearly relaxed-regular.

Let's analyze possible ways of reaching a configuration from $L_1$.
Surely $\s{G}^1$ and $\s{G}^4$ cannot fire the disappearing transitions, because their presence is required by $L_1$.
As $\s{G}^1$ and $\s{G}^4$ are always top-most on their stacks, all other symbols on these stacks
are 'guarded' -- they can not fire a disappearing transition neither.
A key observation is that no symbol from other two stacks could fire a disappearing transition:
%Intuitively, a disappearing transition would result in either $\alpha_2$ smaller than $\alpha_1$, or
%$\alpha_3$ smaller than $\alpha_4$, according to the subword ordering.
\begin{lem}
Every path from the initial configuration to $L$ contains no disappearing transitions.
\end{lem}
\begin{proof}
The precise proof of this fact needs a certain effort. 
Let us define the \emph{weight} of a nonterminal.
The intuition behind this notion is that it counts for how many letters in words $s_i$ the particular nonterminal is responsible.
The definition is the following:
\begin{iteMize}{$\bullet$}
  \item $\size(\s{G}^1) = \size(\s{G}^4) = 0$
  \item $\size(i^1) = \size(i^2) = \mbox{length}(s_i)$
  \item $\size(a^3) = \size(a^4) = 1$
\end{iteMize}
Weight of a word is defined as the sum of weights of its letters.
Note now that any configuration $\alpha = (\s{G}_1 \, \alpha_1, \ \alpha_2, \ \alpha_3, \ \s{G}_4 \, \alpha_4)$
reachable from $(\s{G}_1, \ \eps, \ \eps, \ \s{G}_4)$ satisfies the following inequalities:
\begin{align*}
\mbox{SInv}_1(\alpha) & \ \ = \ \ \size(\alpha_1) - \size(\alpha_3)   \ \ \geq \ \ 0 \\
\mbox{SInv}_2(\alpha) & \ \ = \ \ \size(\alpha_4) - \size(\alpha_2)   \ \ \geq \ \ 0.
\end{align*}
To see this it is enough to observe this both semi-invariants $\mbox{SInv}_1$ and $\mbox{SInv}_2$ equal $0$ in the initial configuration and 
that they never decrease due to performing a transition.
In particular, every disappearing transition on the second or third stack increases one of the semi-invariants.
Finally, every configuration $\alpha \in L$ satisfies the equality:
\[
\mbox{SInv}_1(\alpha) + \mbox{SInv}_2(\alpha) = 0,
\]
as $\size(\alpha_1) = \size(\alpha_2)$ and $\size(\alpha_3) = \size(\alpha_4)$. 
Therefore both semi-invariants are necessary equal to $0$, and thus there is no possibility
for disappearing transitions to be fired.
\end{proof}
As a conclusion we obtain:
\begin{cor}
Consider any configuration in $L_1$ that is reachable from the initial configuration,
and suppose that its first and forth stacks have the form:
\begin{align*}
\s{G}^1 \, i_1^1 \ldots i_k^1 &&
\s{G}^4 \, a_1^4 \ldots a_m^4 .
\end{align*}
Then it holds:
\begin{align*}
s_{i_1} \, \ldots \, s_{i_k} \ \  = \ \  a_1 \ldots a_m .
\end{align*}
\end{cor}

\smallsection{Complete solution.}
Similarly as above, one may construct an \MPDA\ $\aut{A}_2$ and a language $L_2$ for the right-hand side words $t_i$ 
from the PCP instance. Essentially (i.e., ignoring the technical details) 
the reachable configurations of $\aut{A}_2$ intersected with $L_2$ are
(cf.~\eqref{eq:idea}):
\begin{align*}
(i_1 \, i_2 \ldots i_k, \ \ t_{i_1} \, t_{i_2} \ldots t_{i_k}).
\end{align*}
% The words $t_i$ are represented by nonterminals $\s{T}_i^1$ and $\s{T}_i^2$, similarly as the words $s_i$ were represented above.
Our final solution is to appropriately combine both \MPDAs\ and both languages.

The \MPDA\ $\aut{A}$ is obtained by merging $\aut{A}_1$ and $\aut{A}_2$, but the first stacks are identified.
Thus $\aut{A}$ will have seven stacks altogether.
In particular, symbols $i^1$ and $i^2$ represent now the $i$th pair $(s_i, t_i)$.
All transitions are exactly as described above, however with a different numbering of stacks.
The language $L$ imposes the requirements of $L_1$ and $L_2$, and additionally requires that 
the fourth stack of $\aut{A}_1$  is almost equal to the fourth stack of $\aut{A}_2$.

For describing the missing details we have to fix a new numbering of stacks. Let the first four stacks
correspond to the stacks of $\aut{A}_1$, and the remaining three stacks correspond to the stacks of $\aut{A}_2$
different than the first one.
The initial configuration of $\aut{A}$ is
\[
(\s{G}^1, \ \eps, \ \eps, \ \s{G}^4, \ \eps, \ \eps, \ \s{G}^7) .
\]
Except for the disappearing transitions, $\aut{A}$ has  the following transition rules:
\begin{align*}
\s{G}^1 \ & \trans \ \s{G}^1 \, i^1, \ \eps, \ s_i^3, \  \eps, \ \eps, \ t_i^6, \  \eps   \\
\s{G}^4 \ & \trans \ \eps, \ i^2, \ \eps, \ \ \s{G}^4 \, s_i^4, \ \eps, \ \eps, \  \eps \\
\s{G}^7 \ & \trans \ \eps, \ \eps, \ \eps , \ \eps, \ i^5,  \  \eps, \ \ \s{G}^7 \, t_i^7 .
\end{align*}
The language $L$ contains configurations of the form:
\[
(\s{G}^1 \, \alpha_1, \ \alpha_2, \ \alpha_3, \ \s{G}^4 \, \alpha_4, \ \alpha_5, \ \alpha_6, \ \s{G}^7 \, \alpha_7)
\]
satisfying the following almost equalities:
\[
\alpha_1 \ = \ \alpha_2 = \alpha_5 \qquad
\alpha_3 \ =  \ \alpha_4 = \alpha_6 \ = \ \alpha_7.
\]
One can easily observe that $L$ is reachable from the initial configuration
if and only if the PCP instance has a solution, using exactly the same techniques as before.
\end{proof}

\section{Conclusions} \label{sec:conc}

In this paper we have investigated decidability and complexity of the reachability problem for multi-stack pushdown automata.
As the model is undecidable in general, we have mostly focused on the case of weak control states.
We have provided an almost complete map of complexity results for the reachability problem, depending on
the natural restrictions of source and target sets of configurations, and on \normedness\ assumption.
Few remaining open cases are listed below:
\begin{iteMize}{$\bullet$}
\item The most interesting open question is the exact complexity of \singsing\ reachability problem for 
(\normed\ and \unnormed) \weak\ \MPDAs. We only know that the problem is NP-hard. Note that the \regsing\ reachability has the
same complexity.
\item Another interesting question is whether three stacks are necessary for undecidability of \singreg\ reachability problem
for \normed\ \weak\ \MPDA. In other words, we do not know the status of the problem when there are just two stacks.
\item A similar question is still open for the cases when reachability is in \NP: is the problem still \NP-complete if one restricts the number of stacks to two?
\item We do not know exact complexity of reachability for \strnormed\ \MPDA.
\item Finally, we do not know whether the construction of backward reachability set of a regular set for strongly normed \MPDAs,
even in the weak case, could be performed effectively.
\end{iteMize}

%%% Local Variables: 
%%% TeX-master: "main"
%%% End: 

\smallsection{Acknowledgements.}
We are grateful to anonymous reviewers for careful reading and many valuable comments.

\bibliographystyle{plain}
\bibliography{citat}

\end{document}

%% file: defs.tex
\theoremstyle{plain}
\theoremstyle{plain}
\theoremstyle{plain}
\theoremstyle{plain}
\theoremstyle{plain}

\newcommand{\undecid}{undecidable}
\newcommand{\decid}{decidable}
\newcommand{\NP}{NP}
\newcommand{\NPc}{NP-compl.}
\newcommand{\sing}{1}
\newcommand{\reg}{{\sc reg}}
\newcommand{\singreg}{\sing{$\reaches$}\reg}
\newcommand{\singsing}{\sing{$\reaches$}\sing}
\newcommand{\regsing}{\reg{$\reaches$}\sing}
\newcommand{\regreg}{\reg{$\reaches$}\reg}

\newcommand{\stateless}{stateless}

\newcommand{\weak}{weak}

\newcommand{\strnormed}{strong\-ly nor\-med}
\newcommand{\strnormedness}{strong nor\-med\-ness}
\newcommand{\normed}{nor\-med}
\newcommand{\normedness}{nor\-med\-ness}
\newcommand{\unnormed}{un\-nor\-med}
\newcommand{\reaches}{\leadsto}
\newcommand{\aut}[1]{{\mathcal #1}}
\newcommand{\MPDA}{MPDA}
\newcommand{\MPDAs}{MPDAs}
\newcommand{\s}[1]{\text{#1}}
\newcommand{\al}[1]{\mathsf{#1}}
\renewcommand{\c}[1]{\mathit{#1}}
\newcommand{\confset}{\al{S}}
\newcommand{\confless}[3]{\langle #1, #2, #3 \rangle}
\newcommand{\config}[4]{\langle #1, #2, #3, #4 \rangle}

\newcommand{\trans}[1][]{\xrightarrow{#1}}     % \stackrel{#1}{\longrightarrow}}
\newcommand{\eps}{\varepsilon}

\newcommand{\Nat}{\mathbb{N}}
\newcommand{\size}{\mbox{weight}}

\renewcommand{\a}{\alpha}
\renewcommand{\b}{\beta}

\newcommand{\g}{\gamma}

\newcommand{\smallsection}[1]{\paragraph{\bf #1}}

\newcounter{examplecounter}
\newenvironment{proofof}[1]{\bigskip \noindent{\bf Proof of #1.}}{\hfill \qed}

\newcommand{\restate}[2]{\noindent {\bf #1.}\ #2 \smallskip}